\documentclass[cmp]{svjour}

 \pdfoutput=1

\usepackage{amssymb}
\usepackage{bm}
\usepackage{amsmath}
\usepackage{amsfonts}
\usepackage{mathrsfs}
\usepackage{units}
\usepackage{cite}

\newcommand{\p}{\partial}

\newcommand{\dd}{{\rm d}}

\newcommand{\inner}[2]{\langle #1, #2 \rangle}
\newcommand{\laplace}{\Delta}
\newcommand{\hodgestar}{\star}
\newcommand{\harm}{\mathcal{H}}

\begin{document}

\title{Compact Cauchy horizons admit constant surface gravity}

\author{R. A. Hounnonkpe\thanks{Universit\'e d'Abomey-Calavi, Abomey-Calavi, B\'enin and Institut de Math\'ematiques et de Sciences Physiques (IMSP), Porto-Novo, B\'enin. ORCID ID: 0000-0002-5013-8329
 \email{rhounnonkpe@ymail.com}} and E. Minguzzi\thanks{Dipartimento di Matematica, Universit\`a degli Studi di Pisa,  Largo
B. Pontecorvo 5,  I-56127 Pisa, Italy. ORCID ID: 0000-0002-8293-3802 \email{ettore.minguzzi@unipi.it}}}

\institute{}

\date{}

\maketitle

\begin{abstract}
\noindent
We prove that in any spacetime dimension and under the null energy condition, every totally geodesic connected smooth compact null hypersurface (hence every compact Cauchy horizon) admits a smooth lightlike tangent vector field of constant surface gravity. That is, we solve the open degenerate case by showing that, if there is a complete generator, then there exists a smooth future-directed geodesic lightlike tangent field. The result can be stated as an existence result for a particular cohomological equation. The proof uses elements of ergodic theory, Hodge theory and Riemannian flow theory. We emphasize that, remarkably, these results really require only the null energy condition, whereas previous works assumed, already in the Killing or the non-degenerate cases, the stronger dominant energy condition.
\end{abstract}

%\keywords{null hypersurfaces, compact horizons, surface gravity}
%\subclass{MSC codes}

\section{Introduction}

The study of closed lightlike geodesics has long posed a fundamental and still unresolved problem in general relativity.

Consider a closed, affinely parametrized lightlike geodesic segment \(\gamma: [0,\lambda] \to M\), $t\mapsto \gamma(t)$, with \(\gamma(\lambda) = \gamma(0)\) and \(\dot{\gamma}(\lambda) = a \dot{\gamma}(0)\) for some \(a > 1\). As observed by Hawking \cite{hawking66b,hawking73}, the future-inextendible geodesic \(\tilde{\gamma}\) winds around \(\gamma\) infinitely many times. On the first loop, it accumulates an affine length \(\lambda\); on the second, \(\lambda/a\); and so on. The total affine length is therefore \(\lambda \sum_{i=0}^\infty (1/a)^i = \lambda \frac{a}{a-1}\), which is finite. A similar argument in the backward direction shows that the inextendible lightlike geodesic covering \(\gamma\) is future incomplete but past complete.

In this future-incomplete case, the lightlike vector field
\[
n(t) = \left(1 - \tfrac{t}{\lambda}\left(1 - \tfrac{1}{a}\right)\right) \dot{\gamma}(t)
\]
is well-defined over the entire loop and satisfies the pregeodesic equation
\[
\nabla_n n = \kappa n, \qquad \kappa = -\tfrac{1}{\lambda}\left(1 - \tfrac{1}{a}\right),
\]
where \(\kappa\) (the surface gravity) is a negative constant.

For \(0<a < 1\), the situation is analogous: the inextendible lightlike geodesic covering \(\gamma\) is past incomplete but future complete, and a future-directed pregeodesic field \(n\) can be found with a constant positive \(\kappa\) in the pregeodesic equation.

Only in the case \(a = 1\) is the inextendible lightlike geodesic complete in both directions. Here, we may take \(n = \dot{\gamma}\), resulting in a zero coefficient in the pregeodesic equation (i.e., \(n\) is geodesic).

The central question is: {\em Does this picture hold for lightlike geodesics constrained to a compact null hypersurface?} Specifically, is it still possible to find a field \(n\) satisfying \(\nabla_n n = \kappa n\) with \(\kappa\) constant everywhere on the hypersurface? If so, it follows readily that all geodesics on the null hypersurface must share the same (in)completeness properties, with these properties determined by the sign of \(\kappa\) as above.

Of course, the general case is far more complex. The generators of the null hypersurface need not close, meaning the hypersurface may exhibit highly intricate dynamics. Nevertheless, as we shall see, under the suitable energy conditions and for compact Cauchy horizons, the answer to this question is affirmative.

Understanding the completeness properties of geodesics on compact null hypersurfaces (and their neighborhoods) is of paramount importance. For instance, it could ultimately lead to the removal or relaxation of the null genericity condition in singularity theorems such as those of Hawking and Penrose—a condition that has faced physical criticism \cite[Remark 6.22]{minguzzi18b}.

The study of compact Cauchy horizons gives  also insights into stationary black holes, as their event horizon can be compactified via a standard trick \cite{friedrich99}. In fact, the compactification strategy has already served several purposes, from the proof of the regularity of the event horizons of stationary black holes, to the proof of Hawking’s local rigidity theorem \cite{petersen19,minguzzi24}.

The  problem of establishing the constant surface gravity of compact horizons is also motivated by broader considerations. The {\em strong cosmic censorship}  conjecture (SCC), originally proposed by Penrose \cite{penrose69,penrose79}, asserts that for generic initial data in general relativity, the maximal globally hyperbolic development of a spacetime cannot be extended. This principle encapsulates the idea that, under typical physical conditions, the future evolution of spacetime is uniquely determined by its past. The existence of exact solutions containing Cauchy horizons is understood as a consequence of their non-generic nature. If SCC holds, the formation of compact Cauchy horizons must be highly exceptional, occurring only under very special circumstances.

A precise formulation of this expectation was given by Isenberg and Moncrief \cite{moncrief83}, who conjectured that any spacetime admitting a non-empty compact Cauchy horizon must exhibit a Killing symmetry—making it non-generic—and that the horizon itself must be a Killing horizon. Their initial proof established this result under three key assumptions:
\begin{itemize}
\item[(a)] the metric and horizon are analytic,
\item[(b)] Einstein's (electro-)vacuum equations are satisfied,
\item[(c)] the horizon generators are closed.
\end{itemize}

Subsequent progress has relaxed these restrictions. Notably, it was shown that smoothness (analyticity) of the horizon follows from the smoothness (resp. analyticity)  of the metric when the null energy condition holds \cite{larsson14,minguzzi14d}, rendering the assumption on the regularity of the horizon in (a) unnecessary. Under this energy condition, compact Cauchy horizons are smooth, totally geodesic null hypersurfaces, which has become the standard setting for further analysis.

The totally geodesic property of such horizons can be expressed by the equation
\begin{equation} \label{pprt}
\nabla_X n = \omega(X) n,
\end{equation}
where \(X\) is any vector field tangent to the horizon \(H\), \(n\) is a future-directed lightlike generator, and \(\omega: H \to T^*H\) is a 1-form on \(H\). Here, \(\nabla\) may be interpreted equivalently as the spacetime Levi-Civita connection or the induced connection on \(H\), with no ambiguity due to the totally geodesic condition. The {\em surface gravity} \(\kappa\) is defined as the contraction \(\kappa := \omega(n)\). Under a rescaling \(n' = e^{-f} n\) of the generator, the 1-form and surface gravity transform as
\begin{align}
    &\omega' = \omega - \mathrm{d}f, \label{ren1} \\
    &\kappa' = e^{-f} (\kappa - n(f)). \label{ren2}
\end{align}

A central question is whether, under the dominant energy condition, \(\kappa\) can be normalized to a constant with a suitable choice of $n$, which is a prerequisite for the horizon to be Killing \cite{bardeen73,chrusciel20}. Establishing constant surface gravity
%further implies the vanishing of the Lie derivative \(L_n g = 0\) along the horizon, opening the door to
might thus open the way to
extending \(n\) to a Killing field via PDE techniques.

It should be noted that if $n$ can be found of constant surface gravity $\kappa$ over a {\em compact} horizon, then
\[
\nabla_n n =\kappa n,
\]
 and under this equation, taking advantage of the completeness of $n$ (by standard ODE theory) it is possible to show that   the generators  are necessarily all affine complete for $\kappa =0$ or all affine incomplete for $\kappa \ne 0$.

Indeed, with the objective of proving constant surface gravity, compact horizons have been  classified as {\em degenerate} (if at least one generator is complete) or {\em non-degenerate} (if at least one generator is incomplete). Isenberg and Moncrief demonstrated a {\em dichotomy} under their original assumptions:  generators are either all complete (degenerate case) or all incomplete (non-degenerate case) (in which case they are all future-incomplete or all past-incomplete)  and ultimately proved that compact horizons are Killing (and, in particular, of constant surface gravity) \cite{moncrief83}.

Over four decades, their collaborative work \cite{moncrief83,isenberg85,isenberg92,moncrief08,moncrief20} progressively weakened assumptions (a) and (c), introducing innovative tools such as {\em Gaussian null coordinates} and the {\em ribbon argument}. Recent breakthroughs by Bu\-sta\-mante–Rei\-ris and by Gurriaran and the second author \cite{minguzzi21} established the validity of the dichotomy in the smooth category and under the dominant energy condition. Further, the non-degenerate case was solved in full generality, proving the constancy of \(\kappa\) (with \(\kappa \neq 0\)) under the said assumptions and in any dimensions.
 Under the additional vacuum assumption, Petersen and R\'acz  \cite{petersen18,petersen18b,petersen19} then completed the extension of \(n\) to a Killing field, settling the Isenberg–Moncrief conjecture for non-degenerate horizons.

We can already improve the previous results by observing that constancy of surface gravity in the non-degenerate case (and, as we shall see, also in the degenerate case), holds already under the null energy condition.

The identity for totally geodesic null hypersurfaces \cite[Lemma 3]{minguzzi21}
\begin{equation} \label{nnrt}
Ric(n,X)=\dd \omega(n,X), \quad X\in TH,
\end{equation}
will be useful. Using the dominant energy condition it is possible to derive the condition $Ric(n,\cdot)\vert_{TH}=0$ (here $\vert_{TH}$ denotes restriction to the subbundle $TH$) which, due to (\ref{nnrt}), is equivalent to
\[
i_n\dd \omega=0,
\]
 see \cite[Lemma 4-5]{minguzzi21}. In fact, the dominant energy condition in \cite{minguzzi21} is used just to derive this equation and for no other implications.
 %Jointly with the Riemannian flow condition $L_Tg=0$, which will be recalled in a moment, this is the used to apply the ribbon argument, est

%  The last equation is really all that is needed for the application of the ribbon argument, and hence for  the dichotomy and  the constancy of surface gravity in the non-degenerate case to hold.

%  following from the totally geodesic property, implies $Ric(n,n)=0$.

However, we have the following simple algebraic fact (we recall that a lightlike vector is a null non-zero vector) which proves that the null energy condition is sufficient to get the same result.

\begin{proposition} \label{pyrt}
Let  $(M,g)$ be a spacetime and let $n \in T_p M$ be a lightlike vector at $p\in M$ such that $R(n,n)=0$ where $R: T_pM\times T_pM\to \mathbb {R}$ is a symmetric bilinear form  at $p$. If $R(v,v)\ge 0$ for every null vector $v\in T_pM$ then  $R(n,\cdot)=b g(n,\cdot)$ for some constant $b\in \mathbb{R}$.
\end{proposition}

\begin{proof}
Let $N$ be a lightlike vector such that $g(N,n)=-1$, and let $u$ be an arbitrary vector such that $g(n,u)=g(N,u)=0$, hence spacelike.
The vector $v=n+xu+ x^2\frac{1}{2} g(u,u) N$ is lightlike and $R(v,v)=2 x R(n,u) + \ldots x^2+\ldots x^3+\ldots x^4$ is a polynomial of fourth-degree whose dominant term for small $x$ is linear. Since $x$ can be chosen arbitrarily small and of any desired sign, the condition $R(v,v)\ge 0$ implies $R(n,u)=0$. Since we already have $R(n,n)=0$ we conclude that for every $w\in \ker g(n,\cdot)$ we have $R(n,w)=0$.  Since $\ker R(n,\cdot)\subset \ker g(n,\cdot)$ there is $b\in \mathbb{R}$ such that $R(n,\cdot)=b g(n,\cdot)$ (in fact $b=R(n,N)/g(n,N)$).  $\hfill$ \qed
\end{proof}

\begin{corollary} \label{cooe}
Assume the null energy condition. Over any smooth totally geodesic null hypersurface $H$, if $n\in TH$ is lightlike, we have $Ric(n,\cdot)\vert_{TH}=0$ and hence $i_n \dd \omega=0$.
\end{corollary}

\begin{proof}
Letting $X=n$ in Eq.\ (\ref{nnrt}) we get $Ric(n,n)=0$, thus by the corollary  $Ric(n,\cdot)\vert_{TH}=0$, which, again due to  Eq.\ (\ref{nnrt}), is equivalent to $i_n \dd \omega=0$. $\hfill$ \qed
\end{proof}

\begin{theorem} \label{nphb}
All the results in \cite{minguzzi21} that use the dominant energy condition actually hold under the weaker null energy condition.
\end{theorem}

In particular, the dichotomy holds \cite[Cor.\ 2]{minguzzi21} and so we can make the distinction between degenerate and non-degenerate horizons also under this energy condition.

\begin{remark}
Previous derivations of the constancy of surface gravity for Killing horizons assumed the dominant energy condition \cite{bardeen73,chrusciel20} precisely to derive the equation  $Ric(n,\cdot)\propto g(n,\cdot)$, see \cite[Prop.\ 4.3.11]{chrusciel20}.
Of course, in this Killing context, and regardless of non-degeneracy and totally geodesic\footnote{The totally geodesic property can be derived for a null hypersurface $H$ admitting a tangent Killing field $k$. Indeed, if $X,Y: H \to TH$, $g(\nabla_X k,Y)=-g(k, \nabla_XY)$ and using the Killing property $g(\nabla_X k, Y)=-g(\nabla_Yk, X)$, thus summing $2 g(\nabla_X k, Y )=g(k, [Y,X])=0$ which implies $\nabla_X k \propto k$ on $H$.} assumptions, the energy requirement can be weakened to the null energy condition for the reasons stated in Prop.\ \ref{pyrt} as the equation $Ric(n,n)=0$ can also be derived, see \cite[Eq.\ (4.3.37)]{chrusciel20}. Although the derivation of surface gravity constancy in the Killing case is well-established, this potential improvement has apparently gone unnoticed.
\end{remark}

Let us come to the degenerate case. So far progress in the smooth category has been scarce.
However, in a recent work \cite{minguzzi24c} we pointed out the usefulness of Riemannian flow theory and the fact that its implications apply equally well to degenerate  and non-degenerate cases.  In fact, we were able to obtain a classification result for the possible structures of the horizon. In order to better introduce these findings, we need to be more precise about  the definitions.

For brevity, we adopt the following terminology:

\begin{definition}
A {\em horizon} is a  smooth connected, totally geodesic  null hypersurface. %of class \( C^r \), $r\ge 3$.
\end{definition}

In this work, a horizon will be denoted with \( H \), and \( n \) will denote a future-directed null vector field tangent to \( H \).

The achronality property of $H$ will not be important for what follows and so will not be imposed. Anyway,
 local achronality can be proved by using the existence of convex neighborhoods, and then the spacetime can be reduced to an open neighborhood of $H$ so as to get achronality if desired \cite[Prop.\ 4]{minguzzi21}.

While our primary focus is on compact horizons,
compactness will be explicitly stated where necessary.

The induced metric on the horizon is denoted \( g_T \) and termed the {\em transverse metric}. Abstracting away from the spacetime, a horizon can be characterized as a triple\footnote{For a more general abstract approach to the geometry of null hypersurfaces, see \cite{manzano24}. For the rigidity aspects of null geometry, see \cite{bekkara06}.} \( (H, g_T, \nabla) \), where:
\begin{itemize}
\item[-] \( g_T \) is a positive semi-definite symmetric bilinear form with a 1-dimensional oriented kernel \( \text{span}(n) \),
\item[-] \( n \) is a smooth vector field (future-directed) satisfying \( g_T(n, \cdot) = 0 \),
\item[-] \( \nabla \) is a symmetric affine connection compatible with \( g_T \) and \( n \) via \( \nabla g_T = 0 \) and \( \nabla n = n \otimes \omega \) for some 1-form \( \omega \).
\end{itemize}
These conditions imply that the flow of \( n \) preserves \( g_T \) (cf. \cite[Thm. 37]{kupeli87}, \cite[Lemma B1]{friedrich99}, \cite{moncrief08}, \cite[Lemma 7]{minguzzi21}), in the sense that for $X,Y\in TH$:
\begin{equation}
L_n g_T(X,Y) = 0.
\end{equation}
The kernel \( \text{span}(n) \), being 1-dimensional, is automatically integrable, with its integral curves forming the (unparametrized) generators of the horizon.

We now recall a separate concept.

A {\em foliation} \( (H, \mathcal{F}) \) on a manifold consists of a completely integrable distribution \( D \) of constant dimension. The integral leaves are denoted \( L \), and the collection of leaves is \( \mathcal{F} \). Two foliations \( (M, \mathcal{F}) \) and \( (M', \mathcal{F}') \) are {\em conjugate} if a diffeomorphism \( M \to M' \) maps leaves to leaves.

If a manifold \( H \) admits a positive semi-definite symmetric bilinear form \( g_T: H \to T^*H \otimes_H T^*H \) with a kernel
\[
\ker g_T = \{ X \in TH : g_T(X, \cdot) = 0 \}
\]
of constant dimension, and \( L_X g_T = 0 \) for all \( X \in \Gamma(\ker g_T) \), then \( \ker g_T \) is integrable and defines a foliation \( \mathcal{F} \) on \( H \) \cite[Prop. 3.2]{molino88}.

Such a foliation \( (H, \mathcal{F}) \) is called {\em Riemannian} if it arises from a metric \( g_T \) with the above properties \cite[Sec. 3.2]{molino88}. The metric $g_T$ need not be unique. When there is a privileged choice, the pair \( (H, g_T) \) itself might be referred to as a Riemannian foliation.

A {\em Riemannian flow} is a Riemannian foliation of oriented 1-dimensional leaves \cite{carriere84}. We can denote with \( n \)  a non-vanishing, positively oriented vector field tangent to the foliation, unique up to positive scaling. The Riemannian flow condition requires \( L_n g_T\vert_{TH\times TH} = 0 \) for some choice of \( n \), since the equation  \( L_{fn} g_T\vert_{TH\times TH} = 0 \) for any positive function $f$ can be deduced \cite{minguzzi24c}.
Not all oriented 1-dimensional foliations admit such a metric \cite{molino88}.

The following proposition connecting the notion of horizon with that of Riemannian flow is not hard to prove \cite{minguzzi24c}:

\begin{proposition} \label{bort}
Every horizon naturally defines a Riemannian flow with a prescribed \( g_T \), the leaves being the unparametrized generators, and $g_T$ being induced from the spacetime metric.
Conversely, any smooth null hypersurface \( H \) inducing a Riemannian flow structure \( (H, g_T) \) must be totally geodesic.
\end{proposition}

The orientation of the flow comes, of course, from the orientation of the spacetime.

In short, we can regard each compact horizon as a Riemannian flow with a canonically defined transverse metric $g_T$.

We recall that a Riemannian flow is {\em isometric} if there exists a Riemannian metric $h$ and a choice of $n$ on $H$ such that $L_n h=0$. In fact the pair $(n,h)$ can be chosen so that $h(n,n)=1$ without loss of generality \cite[Remark 3]{minguzzi24c}. Not all Riemannian flows are isometric. In \cite{minguzzi24c} we presented several equivalent characterizations of this important property.

The theory of Riemannian foliations began with the work by Reinhart \cite{reinhart59} (1959) and systematic treatments have appeared in specialized books, including those by Molino \cite{molino88} and Tondeur \cite{tondeur97}. Reinhart \cite{reinhart59b} also introduced the {\em basic cohomology}. For a Riemannian flow a $k$-form $\alpha$ is basic if $i_n \alpha=0$ and $i_n\dd \alpha=0$. In other words, basic forms arise as pullback forms of local quotient manifolds (the manifold can be covered by cylinders that admits a quotient under the flow).  Note that a 0-th form $f$ is basic if $n(f)=0$, i.e.\ it is constant over the leaves.

The exterior differential in basic cohomology coincides with the standard exterior: note that if $\alpha$ is basic then $\dd \alpha$ is basic. The cohomological space $H^k_b(H)$ is the quotient of the space of closed basic forms with that of exact basic forms.

\begin{remark} \label{crre}
If two basic 1-forms $\alpha$ and $\beta$ differ by an exact form $\beta-\alpha =\dd f$ then, as contracting with $n$ shows, $n(f)=0$, which proves $i_n \dd f$ and hence that the 0-th form $f$ is basic, as we have also $i_n f=0$. Thus $\alpha$ and $\beta$ differ also by the exterior differential of a basic form. In other words, we have the following result relating de Rham cohomology and  basic  cohomology:  $H^1(H)=0$ implies $H^1_b(H)=0$.
\end{remark}

An important result by Carri\`ere establishes that for a Riemannian flow the closure of the orbits are tori \cite{carriere84}. The maximal dimension for a closed tori in the manifold is denoted $k+1$.

By drawing from results by Molino and Carri\`ere \cite{carriere84} we obtained a classification for the  structures of a compact horizon in a 4-dimensional spacetime \cite[Thm.\ 14]{minguzzi24c}. On this Riemannian flow structure results conjugacy is always understood in the orbital sense, i.e.\ without parameter, meaning that orbits are sent to orbits but their parametrization might not be preserved.

\begin{theorem} \label{car}
Let $(M,g)$ be a  smooth, $4$-dimensional spacetime and  let $H$ be  an orientable compact horizon in $M$. There are the following mutually excluding possibilities:
\begin{enumerate}
\item[(i)] All orbits are closed, $k=0$. $H$ is a Seifert fibration over a Satake manifold. The  fibers are the orbits of the flow.
\item[(ii)] The flow has precisely two closed orbits and every other orbit densely fills a two-torus, $k=1$. If the closed orbits are removed the manifold is diffeomorphic to a product $(0,1) \times  T^2$ where the flow in each fiber $\{s\} \times T^2$ is a linear flow with dense orbits. The manifold $H$ is obtained by gluing two solid tori by their boundaries, on each solid torus the foliation being given through a suspension by an irrational rotation of the disk. Moreover, there are two subcases:
\begin{itemize}
\item[a)] $H$ is diffeomorphic to the lens space\footnote{It comprises the case $L(1,0)=S^3$ while the case $L(0,1)=S^1 \times S^2$ is included in (ii-b). See \cite{saveliev12} for this gluing construction. Note that a $S^3$ topology can also appear in case (i) (e.g.\ Hopf bundle).} $L(p,q)$,  and the flow is conjugate to that of \cite[Example 1]{minguzzi24c}.
\item[b)] $H$ is diffeomorphic to $S^1\times S^2$ and the flow is conjugate to the flow given by the suspension of an irrational rotation of $S^2$ (with respect to, say, the $z$-axis in the standard isometric embedding in $\mathbb{R}^3$).
\end{itemize}
\item[(iii)] Every orbit densely fills a two-torus, so no orbit is closed, no orbit is dense,  $k=1$. $H$ is a $T^2$-bundle over $S^1$, each  torus fiber being the closure of the orbits in it, and there are two possibilities:
\begin{itemize}
\item[a)] $H$ is diffeomorphic to $T^3$ with a flow conjugate to a linear flow on the torus.
\item[b)] $H$ is diffeomorphic to the hyperbolic fibration $T^3_A$, ($
\textrm{tr} A > 2$) and the flow is conjugate to one of the flows of \cite[Example 2]{minguzzi24c}.
\end{itemize}
\item[(iv)] All orbits are dense, $k=2$. $H$ is diffeomorphic to $T^3$ with a flow conjugate to a linear flow on the torus.
\end{enumerate}
Moreover, except case (iii-b), the flow is isometric.
\end{theorem}

This result applies independently of whether the horizon is degenerate or non-degenerate. Under the dominant (actually null, as we have shown previously in Thm. \ref{nphb}) energy condition non-degenerate horizons are  isometric \cite{petersen18b,bustamante21}\cite[Cor.\ 7]{minguzzi24c}, and so case (iii-b) is excluded. This also leads to  the classification obtained by Bustamante–Reiris for the non-degenerate case via Lie group methods (see also \cite{kroenke24} for related results in the analytic case).

For completeness we mention the classification for spacetime dimension 2 and 3.

\begin{theorem} \label{carf}
Let $(M,g)$ be a  smooth,  spacetime and  let $H$ be  an orientable compact horizon in $M$. In spacetime dimension 2 the manifold $H$ is diffeomorphic to $S^1$ while in spacetime dimension 3 the manifold $H$ is diffeomorphic to a torus with the flow being conjugate to (a) a constant flow  along the fibers of the second projection  $S^1\times S^1\to S^1$, (b) a linear flow over $T^2$ with dense orbits.
In all cases the flow is isometric.
\end{theorem}

\begin{proof}
The statement for dimension 2 is clear as $H$ has dimension 1. In spacetime dimension 3 the horizon has dimension 2 and by Carri\`ere result, the leaves of a Riemannian foliation have closures which are tori and the Riemannian flow restricted to it is conjugate to a linear flow on the torus with dense orbits \cite[Thm.\ V.A.1]{carriere81}\cite{carriere84}. Thus either there is a non-closed orbit, which implies that its closure and hence $H$ is a torus with the flow as in (b), or all orbits are closed. In the latter case $H$ is obtained by gluing the two $S^1$ sides of a cylinder $S^1\times [0,1]$ which can only be done in one way due to the orientation of the generators. This shows that we are in case (a). Both flows preserve the standard flat metric of the torus, hence the flows are isometric. $\hfill$ \qed
\end{proof}

Despite these advances, the degenerate case remained unsolved as a proof of the constancy of surface gravity (necessarily $\kappa=0$) was lacking. This paper solves this problem by showing that

\begin{theorem} \label{jdot}
Let  $(M,g)$ be a smooth spacetime (of any dimension) and assume the null energy condition. Let $H$ be a  compact degenerate horizon, then there exists  a smooth future-directed lightlike tangent vector field $n$ of zero surface gravity, namely geodesic: $\nabla_n n=0$.
\end{theorem}

Joining with the results of \cite{reiris21,minguzzi21}

\begin{corollary}
Let  $(M,g)$ be a smooth spacetime (of any dimension) and assume the null energy condition. Let $H$ be a  compact horizon, then there exists  a smooth future-directed lightlike tangent vector field $n$ of constant surface gravity: $\nabla_nn=\kappa n$, with $\kappa =cnst$.
\end{corollary}

The proof is based on elements of ergodic theory, Hodge theory and Riemannian flow theory.

We first provide the proof for the isometric case (Thm.\ \ref{mcrt}). In four dimension, there is only one exceptional non-isometric structure which  we deal with using some special results. In fact we find two ways of dealing with this special case. Ultimately, then, we prove the theorem for 4-dimensional spacetimes and lower dimensions (Thm.\ \ref{mcrf}).

The proof valid for generic dimensions is given later. Actually, we provide two proofs. One uses the notion of Carri\`ere neighborhood while the other use basic cohomology. We note that in previous literature most progress was limited to specific dimensions.

We end this introduction commenting on the problem of the existence of compact degenerate horizons under  a vacuum assumption. As noted in \cite[Remark 5]{minguzzi24c} they exist and can also be isometric.  Isenberg and Moncrief \cite{moncrief20} conjectured that they do not exist if additionally the spacetime is analytic and they are Cauchy horizons \cite[Sec.\ 6.5]{hawking73} \cite[Def.\ 3.14]{minguzzi18b}, $H=H^+(S)$. Here, they seem to require that the hypersurface $S$ is spacelike and achronal, which is a somewhat strong form of Cauchy horizon condition with respect to the conventions in \cite{minguzzi18b}, according to which, taking as spacetime a neighborhood of $H$ one would always get $H=H^+(H)$, see e.g.\  \cite[Cor.\ 3.18]{minguzzi18b}.  In any case, in our study we shall  impose neither analyticity of the metric, nor the Cauchy condition, nor the vacuum assumption, so we do not have reasons to suspect that compact horizons belong to a small family, although results obtained in this work could be used precisely to narrow it down.

The paper is structured as follows. In Section \ref{cmmq}, we provide background on the cohomological equation and explain its role in studying degenerate horizons. In Section \ref{ribbr}, we show that for every horizon and choice of vector field \( n \), one can associate an important constant related to various averaged quantities, and we prove that this constant vanishes for degenerate horizons.

In Section \ref{hodrt}, we introduce Hodge theory and demonstrate its utility in solving the isometric case. The proof of our main theorem for 4-dimensional spacetimes is also given here (Theorem \ref{mcrf}), where we address the special non-isometric case (iii-b) via additional technical arguments.

Section \ref{nnnt} presents the proof for the general case (which also yields an alternative proof for the 4-dimensional setting). The key idea is to cover the horizon with {\em Carri\`ere neighborhoods}—non-compact suspensions that are isometric Riemannian flows. To handle non-compactness, we embed these neighborhoods into larger compact suspensions, allowing us to apply results from the isometric case for Riemannian flows.

Section \ref{mnlc} presents a second proof of the general case. This time the non-isometric  case is dealt with using some results from basic cohomology theory, specifically the Gysin sequence.

Finally, in Section \ref{moetr}, we prove a stronger result for horizons \( H \) that are tori with dense orbits: the vector field \( n \) can be chosen such that \( \omega \) is closed, thus defining an element of the basic cohomology \( H^1_b(H) \).

As for conventions and terminology, we
 adopt the notations of \cite{minguzzi14d,minguzzi18b}. A {\em spacetime} \((M,g)\) is a paracompact, time-oriented Lorentzian manifold of dimension \(n+1 \geq 2\) with signature \((-,+,\dots,+)\). The manifolds $M$ and the metric \(g\) are assumed to be smooth. We do not try to work with metrics of regularity \(C^r\), say $r\ge 3$, because most of the theory of Riemannian flows and Hodge theory on which results we rely has been developed in the smooth category.
 Still we believe that  our main result should hold under much weaker regularity conditions.

\section{The cohomological equation} \label{cmmq}

A non-vanishing vector field $n$ over a manifold $H$ induces a flow $\varphi: \mathbb{R} \times H\to H$ and a foliation $\mathcal{F}$. The distribution $D$ is provided at every point $p\in H$ by $D(p):=\textrm{span}(n)\subset T_pH$. The sections of the bundle $\Lambda^r(D^*)\to H$ form the space $\Omega^r_\mathcal{F}(H)$ of foliated forms of degree $r$.

The foliated cohomology can be defined for more general foliated manifolds but is particularly simple to describe in this flow setting \cite{dehghan07}. Let $\chi$ be a 1-form whose kernel is supplementary at each point to the tangent space $D:=\textrm{span}(n)$ of the leaves and $\chi(n)=1$. Then
\begin{equation}
\Omega^r_\mathcal{F}(H) =
\begin{cases}
   C^\infty(H) & \textrm{if} \ r=0 \\
     C^\infty(H)\otimes \chi\vert_D & \textrm{if} \ r=1 \\
    0 & \textrm{if} \ r\ge 2
\end{cases}
\end{equation}
Since we are interested on the action of these forms on vectors belonging to $D$, the choice of the kernel of $\chi$ is irrelevant. We have the foliated complex
\[
0 \to \Omega^0_\mathcal{F}(H) \xrightarrow{d_\mathcal{F}} \Omega^1_\mathcal{F}(H)\to 0
\]
where $d_\mathcal{F} f = n(f) \chi\vert_D$. Observe that rescaling $n$ with a positive function, rescales also $\chi$ and so $d_\mathcal{F}$ does not depend on the choice of $n$. We are interested in the foliated cohomology group $H^1_{\mathcal{F}}(H)=\Omega^1_\mathcal{F}(H)/d_\mathcal{F} \Omega^0_\mathcal{F}(H)$. An element $\kappa \chi\vert_D \in  \Omega^1_\mathcal{F}(H)$, where $\kappa$ is a function, is exact if there is a function $f$ such that
\begin{equation} \label{scpt}
n(f)=\kappa.
\end{equation}
This is the {\em cohomological equation}. Roughly, $H^1_{\mathcal{F}}(H)$ is the space of non-tri\-vial obstructions to the existence
of a smooth solution $f$ of (\ref{scpt}) for any given smooth $\kappa$.
Observe that for $H$ compact if $\kappa\ne 0$ is a constant then $\kappa \chi\vert_D$ is not exact, indeed  as every orbit accumulates to some point $p$, if it were $n(f)=const\ne 0$ we would get a contradiction with the continuity of $f$ at $p$. This shows that $\textrm{dim} H^1_\mathcal{F}(H)\ge 1$.

\begin{remark} \label{ccor}
In our compact horizon framework let us start by choosing any future-directed lightlike tangent field $n$.
Inspection of Eq.\ (\ref{ren2}) shows that a solution to  the cohomological equation provides a field $n'=e^{-f} n$ of zero surface gravity: $\kappa'=0$.
\end{remark}

The literature on the cohomological equation is extensive but results are focused on proving the existence of $f$ for {\em any} $\kappa$ up to some normalization conditions to be imposed on the integral of $\kappa$. In general, the problem can be solved under restrictive conditions on the flow, in fact Katok's conjecture (now proved for  $H$ of dimension 3 \cite{matsumoto09,kocsard09,forni08}, which is the physical 4-dimensional spacetime case) states that on a connected orientable compact manifold $H$ the equation $n(f)=\kappa - c(\kappa)$ can  be solved for any function $\kappa$ and for some constant $c(\kappa)$, only if the flow is conjugate to a Diophantine linear flow on the torus.\footnote{Here {\em Diophantine} is what is also commonly referred as {\em badly approximable by rationals}. Also {\em conjugacy} here is understood with parameter: A flow $\phi^t: M \to M$ is smoothly conjugate to a flow $\psi^t: N \to N$ if there exists a smooth diffeomorphism $h: M \to N$ such that:
  $ h(\phi^t(x)) = \psi^t(h(x)) \quad \text{for all } x \in M \text{ and all } t \in \mathbb{R}$. Observe that the results by Carri\`ere \cite{carriere84} used elsewhere in this work use orbital conjugacy (i.e.\ parametrization of flow lines is not preserved).}

Instead, in our spacetime problem we need to solve the cohomological equation for just one function $\kappa$. We do not expect severe constraints on the topology of $H$ or on the flow but we have two simplifications on our side. On the one hand, we already know the flow to be Riemannian and, on the other hand, the function $\kappa$ is not any function as it reads $\kappa =\omega(n)$ i.e.\ it is derived from a 1-form $\omega$ with special properties. We shall need to use these facts. In particular, we shall need
that under the null energy condition, on the horizon $H$, $\dd \omega(n,\cdot)=0$ (Cor.\ \ref{cooe})
%
%\cite[Lemma 5]{minguzzi21}
%
%\begin{proposition} \label{cooe}
%Assume the dominant energy condition. Then on the horizon $H$
%\[
%\dd \omega(n,\cdot)=0.
%\]
%\end{proposition}

%A consequence,  which will be interesting when we shall prove that every degenerate horizon admits zero surface gravity, is the following.
%
% Given $n$ future-directed lightlike field $n\in \mathfrak{H}$ we recall that we denote with $\omega$ the 1-form defined by Eq.\ (\ref{pprt}).
%\begin{theorem}
%Assume the dominant energy condition and let $H$ be a horizon.
%\begin{itemize}
%\item[(a)] If $H$ admits zero surface gravity then there is a class $[c]\in \Omega^1_b$ such that for every geodesic  $n$ the 1-form $\omega$ is basic and $[\omega]=[c]$.
%\item[(b)]  If for some $n$ the 1-form $\omega$ is basic, then the horizon admits zero surface gravity.
%\end{itemize}
%\end{theorem}
%In other words, the zero surface gravity property can be identified with the condition $[\omega]=\Omega^1_b$ where the class does not depend on the geodesic field chosen.
%\begin{proof}
%
%\end{proof}

As we shall see, previous studies on the cohomological equation \cite{dehghan07} will really be used just in dealing with the exceptional non-isometric case $T^3_A$ (see case (iii-b) of Thm.\ \ref{car}) as it involves a Diophantine linear flow on the torus.

\section{Consequences of ergodic theory and ribbon argument} \label{ribbr}

We shall need the following theorem which we phrase in the context of Riemannian flow geometry. Indeed, the conclusion only depends on the properties of $\omega$ and on the existence of a transverse metric $g_T$, not on their spacetime origin in our specific problem. It is particularly strong, the proof joining the effectiveness of the ribbon argument with ergodic theory.

%We start with a Lemma following from the results of \cite{minguzzi21}.

We start with a preliminary Lemma  based on the ribbon argument construction discussed in \cite[Sec.\ 4 and App.]{minguzzi21}.
%We refer the reader to that paper for details on the ribbon argument. Here we essentially import result proved there.

\begin{lemma} \label{metq}
Let $H$ be a compact connected manifold endowed with a Riemannian flow generated by a smooth vector field $n$, and let $\omega$ be a 1-form on $H$ which satisfies $i_n\dd \omega=0$.  Let $g_T$ be a transverse metric, i.e.\ $L_n g_T=0$, $g_T(n, \cdot)=0$,  signature $(0,+,\cdots,+)$. If for one integral curve  $x: \mathbb{R}\to H$  of $n$ with starting point $p=x(0)$, the limit
\begin{equation}
    \lim_{t\to \infty} \frac{1}{t} \int_0^t \omega(\dot x) \dd t,
    \end{equation}
exists and vanishes, then the same holds for every integral curve regardless of the starting point. If it exists for two choices of starting points then the corresponding limits have the same sign (if non-zero).
 A similar past version also holds.
\end{lemma}
It can be observed that the argument of the limit is bounded as $\omega(n)$ is a continuous function on  $H$, hence bounded. Thus, if the limit exists, it is finite.

\begin{proof}
%Let $p$ be one such  point and
Suppose that the limit is zero for the starting point $p$.
% and let $q$ any sufficiently close point.
Let $n^*$ be a smooth 1-form field such that  $n^*(n)=1$. A curve with tangent in $\ker n^*$ is said to be $n^*$-horizontal. The limit  $ \lim_{t\to \infty} \frac{1}{t} \int_0^t \omega(\dot x) \dd t$ does not really depend on the starting point chosen along the curve, so for any sufficiently close integral curve $x'$ we can redefine the starting point $x'(0)$ of the parametrization so that there is a horizontal curve $\sigma_0$ connecting $x(0)$ to $x'(0)$. Locally this can be understood working in a cylinder neighborhood \cite[Fig.\ 2]{minguzzi21}. The condition $L_ng_T=0$ ensures that $g_T$ passes to the local quotient manifold to a quotient metric and so we can connect the two (longitudinal) integral curves $x, x'$ with the horizontal lift of a local quotient geodesic.
 It is possible to construct a ribbon from the curves $x\vert_{[0,t]}$,  $x'\vert_{[0,t'(t)]}$,  and two horizontal curves $\sigma_0$ and $\sigma_1$, precisely as done in \cite[Sec.\ 4]{minguzzi21}, and application of Stokes's theorem  implies that \cite{minguzzi21}
\begin{equation} \label{cbti}
\vert \int_{x([0,t])} \omega-  \int_{x'([0,t'])} \omega \vert \le 2B,
\end{equation}
where $B$ is a positive number.
The bound $B$ on the right-hand side does not depend on how elongated is the ribbon (see \cite[Eq.\ (18)-(19)]{minguzzi21}.  In \cite[App.]{minguzzi21} it is detailed how $t\to +\infty$ implies $t'\to \infty$, and conversely. In fact, on the same appendix (footnote 6) it is shown that there is a constant $C>0$ such that $t e^{-C}\le t'\le t e^C$. Dividing the previous equation by $t'$ and taking the limit we get, as $t/t'$ remains bounded, $\lim_{t'\to \infty} \frac{1}{t'} \int_{x'([0,t'])} \omega=0$. This shows that the set $O$ of orbit starting  points for which the integral exists and is equal to zero is open. In fact it is also closed, for if  $q \in \p O$ then for every neighborhood $U\ni q$, we can find some starting point $p\in O$ for which the integral exists and vanishes. Then, reapplying the ribbon argument, we get that the same is true for the orbit from $q$. By the connectedness of $H$ we conclude $O=H$.

The last statement is obtained taking one of the two inequalities implied by Eq.\ (\ref{cbti}), dividing by $t$, and using the boundedness $\log(t/t')$. $\hfill$ \qed
 \end{proof}

\begin{theorem}
\label{Birkm}
%Suppose that the spacetime satisfies the dominant energy condition and let $n$ be any smooth lightlike field tangent to a compact degenerate horizon $H$.

Let $H$ be a compact connected manifold endowed with a Riemannian flow generated by a smooth vector field $n$, and let $\omega$ be a 1-form on $H$ which satisfies $i_n\dd \omega=0$.  Let $g_T$ be a transverse metric, i.e.\ $L_n g_T=0$, $g_T(n, \cdot)=0$,  signature $(0,+,\cdots,+)$. Let $\alpha$ be a 1-form such that $\alpha(n)=1$ and let  $\tilde g=g_T+\alpha \otimes \alpha$. Let $\dd \tilde g$ denote its canonical volume form.

The following properties are equivalent:
\begin{itemize}
\item[(a)] for one  integral curve $x: \mathbb{R}\to H$  of $n$ (and hence for all of them) the following limit exists and is zero
    \begin{equation}
    \lim_{t\to \infty} \frac{1}{t} \int_0^t \omega(\dot x) \dd t,
    \end{equation}
\item[(b)] for one  integral curve $x: \mathbb{R}\to H$  of $n$ (and hence for all of them) the following limit exists and is zero
    \begin{equation}
   \lim_{t\to \infty} \frac{1}{t} \int_{-t}^0 \omega(\dot x) \dd t,
    \end{equation}
\item[(c)]
\begin{equation}
\int_H \omega(n) \dd \tilde g=0,
\end{equation}
\item[(d)] for every compact submanifold $N$ left invariant by the flow (i.e.\ saturated)
\begin{equation}
\int_N \omega(n) \dd \bar g=0.
\end{equation}
where $\bar g$ is the metric induced by  $\tilde{g}$ on the submanifold and $d \bar g$ is the canonical volume form associated to $\bar g$.
\end{itemize}

%
%the existence of one of the limits
%\begin{equation} \label{avrt}
%\lim_{t\to \infty} \frac{1}{t} \int_0^t \omega(\dot x) \dd t, \qquad \lim_{t\to \infty} \frac{1}{t} \int_{-t}^0 \omega(\dot x) \dd t,
%\end{equation}
%and of it being zero implies that the other limit exists and is also zero, and further that it coincides with $\frac{\int_H \omega(n) \dd \tilde g}{\int_H \dd \tilde g}$,
%
%and since the latter does not depend on the orbit, the vanishing of one of the two first integrals is also a property independent of the orbit.
%
%%There is a constant $C\in \mathbb{R}$ such that for every  integral curve $x: \mathbb{R}\to H$ of   $n$
%%\begin{equation} \label{avrt}
%%C=\lim_{t\to \infty} \frac{1}{t} \int_0^t \omega(\dot x) \dd t=\lim_{t\to \infty} \frac{1}{t} \int_{-t}^0 \omega(\dot x) \dd t= \frac{\int_H \omega(n) \dd \tilde g}{\int_H \dd \tilde g}
%%\end{equation}
%%so the limits exist and are independent of the integral curve (and its starting point $x_0=x(0)$).
%
%In this case, for every compact submanifold $N$ left invariant by the flow (saturated)
%\[
%\frac{\int_H \omega(n) \dd \bar g}{\int_H \dd \bar g}=0.
%\]
%where $\bar g$ is the metric induced by  $\tilde{g}$ on the submanifold and $d \bar g$ is the canonical volume form associated to $\bar g$.
\end{theorem}

%Note that $g_T$ and
%NOte $\alpha$ with the property of the theorem exist.
%Also observe that in general the existence of the limit for every orbits

\begin{proof}
%Let $g_T$ be a transverse metric, $L_n g_T=0$, let $\alpha$ be a 1-form such that $\alpha(n)=1$ and let  $\tilde g=g_T+\alpha \otimes \alpha$.
%The properties in parenthesis in (a) and (b) follows immediaately from their weaker versions by Lemma \ref{metq}.

We observe that, using $L_n g_T=0$ and $\alpha(n)=1$, and denoting with $\{e_i\}$ and orthonormal basis of $\ker \alpha $ (see also the proof of \cite[Prop.\ 22]{minguzzi24c})
\begin{align} \label{nqov}
2 \widetilde{\textrm{div}} n&={\textrm{Tr}_{\tilde g}(L_{n} \tilde g)}=(L_{n} \tilde g)(n,n)+\sum_{i=1}^n (L_{n} \tilde g)(e_i,e_i)=2 (L_{n} \alpha)(n)=0.
\end{align}
As a consequence, the flow $(\phi_t)_{t\in \mathbb{R}}$ of $n$ is a measure preserving since it preserves the Riemannian volume element of $\tilde g$.
%, in fact, as shown in the proof of \cite[Prop.\ 22]{minguzzi24c} $\widetilde{\textrm{div}}(n)=0$.
Let us denote $\kappa=\omega(n)$.
Then, Birkhoff's Ergodic Theorem for measure preserving flows gives \cite{birkhoff31,nemytskii60,arnold68}
\[
\int_H \kappa\,d\tilde{g} = \int_H \overline{\kappa}\,d\tilde{g},
\]
where $\overline{\kappa}$ is the function defined almost everywhere by
\[
\overline{\kappa}(p) = \lim_{t\rightarrow +\infty} \frac{1}{t}                   \int_0^{t}\kappa(\phi_s(p))ds.
\]
By Birkhoff's theorem this limit  converges almost everywhere. Assume (a), namely suppose that the limit is zero for some point, then, by Lemma \ref{metq}, the same holds for every point, hence, by Birkhoff's theorem,  $\int_H \omega(n) \dd \tilde g=0$. This shows that (a) implies its strengthened  form and (c).

Assume (c). We know, again by Birkhoff's theorem,  that the limit
\[
\lim_{t\to \infty} \frac{1}{t} \int_{-t}^0 \omega(\dot x) \dd t
\]
 exists almost everywhere. Let us pick two points of the set. Then Lemma \ref{metq} (in the past version) proves that the two limits share the same sign if one of them is non-zero. This means that if there is a point for which the integral is positive (negative) the limit function $\bar \kappa$ is positive  (resp. negative)  almost everywhere thus leading to a non-zero $\int_H \omega(n) \dd \tilde g$. This contradiction proves that one and hence all limits must be zero. This shows that (c) implies (b) (while (b) implies  its strengthened  form by Lemma \ref{metq}).

 % and so actually, by the same argument as before, all the limits of the second type exist for all points and are all zero.

Thus (a) implies (c) which implies (b), and the converse is analogous, so  we have established the equivalence of (a), (b) and (c).

%By the ribbon argument, see Eq.\ (18)-(19) of \cite{minguzzi21}, this limit does not depend on the integral curve $x$ (this conclusion only depends on $L_n g_T=0$ not on the spacetime interpretation of the various objects in that work). Thus we have two constants $C_1= \lim_{t\to \infty} \frac{1}{t} \int_0^t \omega(\dot x) \dd t$ and $C_2=\lim_{t\to \infty} \frac{1}{t} \int_{-t}^0 \omega(\dot x) \dd t$. However, by Birkhoff's Ergodic Theorem applied in the forward and backward directions, both constants are equal to the last space average expression in Eq.\ (\ref{avrt}), thus they are equal.
%Let us consider the closure of an orbit with starting point $x_0$. By Carri\`ere result it is diffeomorphic to a torus $T^a(x_0)$ and the flow is conjugate (without parameter) to a linear flow on the torus. We have shown in  \cite[Prop.\ 24]{minguzzi24c} that there is a measure $\dd \bar g$ preserved by the flow, hence we can apply to the torus the  Birkhoff's Ergodic Theorem concluding that $C_1= \int_{T^a} \kappa \dd \bar g/\int_{T^a}  \dd \bar g$. However, we can apply the theorem also in the backward direction getting $C_2= \int_{T^a} \kappa \dd \bar g/\int_{T^a}  \dd \bar g$, thus the two constant coincide.

Let us prove (a) implies (d). In the proof of \cite[Prop.\ 24]{minguzzi24c} we proved that for $N=H$ or $N$ equal to a torus closure of an orbit, the restricted flow of $n$ is volume preserving.  Actually, the proof applies unaltered for any saturated (i.e.\ union of orbits) submanifold $N$. The Birkhoff's Ergodic Theorem  can thus be applied to the submanifold which proves (d). Conversely, when the last equation of point (d) is applied to $N=H$ we get point (c), thus (d) implies (c). $\hfill$ \qed
\end{proof}

\begin{remark} \label{cnnqa}
Properties (a) and (b), and hence (c) and (d), clearly are not altered changing $\omega$ by an exact form. Also these properties are not altered by the choice of $n$. Indeed, parametrizing the orbit with an arbitrary parameter $s$, it reads $C=\lim_{s\to \infty} \frac{1}{ t(s)}\int_{\gamma([0,s])}\omega$ where $n=\frac{\dd}{\dd t}$. The numerator does not depend on $n$. As for the denominator, if $n'=e^{-f} n$, $\frac{\dd t'}{\dd t}= e^f$ which is bounded from below and above on the compact set $H$, $m\le \vert \frac{\dd t'}{\dd t}\vert \le M$. Thus, $m t(s)\le t'(s)\le M t(s)$ and so $\frac{M}{t'(s)} \ge \frac{1}{t(s)}\ge \frac{m}{t'(s)}$, which shows that if the limit is zero for one choice of $n$ so it is for any other choice.
\end{remark}

%According to the dichotomy proved in \cite[Def.\ 4]{minguzzi21}, a characterization of degenerate horizon is

%Moreover, we established the following consequence of the degenerate assumption\footnote{In this reference the proof is given for $N=H$ and for $N=T^a$, the closure of an orbit, but the proof applies unaltered to any saturated compact submanifold.} \cite[Thm.\ 23, Prop.\ 24]{minguzzi24c}  (we recall that this condition means that one and hence every generator is complete, see the dichotomy in the smooth category as proved in \cite{reiris21}\cite[Cor.\ 2, Def.\ 4]{minguzzi21}) which is valid in any  dimension and will be of use

Returning to a spacetime framework we have the following consequence

\begin{theorem}
\label{Birk app}
Suppose that the spacetime satisfies the null energy condition and let $n$ be any smooth lightlike field tangent to a compact  degenerate horizon $H$.  Let $\omega$ be the 1-form defined by Eq.\ (\ref{pprt}). Then, interpreting $H$ as a Riemannian flow structure, the properties of  Theorem \ref{Birkm} are satisfied.
%Let $\alpha:H\to T^*H$ be a 1-form  such that $\alpha(n)=1$ and let us consider the Riemannian metric $\tilde{g} = g_T + \alpha\otimes \alpha$ on $H$.
%Then for every compact submanifold $N$ left invariant by the flow (saturated)
%\[
%\int_N \kappa\,d\tilde{g} = 0.
%\]
%where $\bar g$ is the metric induced by  $\tilde{g}$ on the submanifold and $d \bar g$ is the canonical volume form associated to $\bar g$.
\end{theorem}

As shown by the proof, the energy condition is only used to get the equality $Ric(n, \cdot)_{TH}=0$ or, equivalently, $i_n\dd \omega=0$.

\begin{proof}
 By Birkhoff's theorem the limit $\lim_{t\to \infty} \frac{1}{t} \int_0^t \kappa(x(t)) \dd t$  exists for some generator. It  cannot be different from zero, for  if it is negative then for the same generator $\int_0^t \kappa(x(t)) \dd t=-\infty$. By the equivalence proved in \cite[Def.\ 4]{minguzzi21} the generators of the horizons would be all future incomplete (the first paragraph of  \cite[Sec.\ 8]{minguzzi21} clarifies that the equivalence in  \cite[Def.\ 4]{minguzzi21}  holds under  $Ric(n, \cdot)_{TH}=0$ which we deduce from the null energy condition, see Cor.\ \ref{cooe} of this work). Similarly, if the limit is positive they would be all past incomplete, thus the limit is zero and so property (a) of Theorem \ref{Birkm} applies. $\hfill$ \qed
\end{proof}
%
%\begin{theorem}
%\label{Birk app}
%Suppose that the spacetime satisfies the dominant energy condition and let $n$ be any smooth lightlike field tangent to a compact degenerate horizon $H$. Let $\alpha:H\to T^*H$ be a 1-form  such that $\alpha(n)=1$ and let us consider the Riemannian metric $\tilde{g} = g_T + \alpha\otimes \alpha$ on $H$.
%Then for every compact submanifold $N$ left invariant by the flow (saturated)
%\[
%\int_N \kappa\,d\tilde{g} = 0.
%\]
%where $\bar g$ is the metric induced by  $\tilde{g}$ on the submanifold and $d \bar g$ is the canonical volume form associated to $\bar g$.
%\end{theorem}
%
%Note that $\alpha$ with the property of the theorem exists.

The proof of the previous theorem uses results
 on the relationship between the behavior of the integral $\int\omega$ over the generators and the affine completeness of the generators, previously obtained in \cite{minguzzi21}. We stress that, as a consequence, the proof of the constancy of surface gravity in the degenerate case will not be independent of results obtained in the non-degenerate case. To the contrary, both the dichotomy and Theorem \ref{Birk app}   depend on recent findings that originated in the study of the non-degenerate case.

\section{Hodge theory and the isometric case} \label{hodrt}

%Introduction to Hodge Theory on Riemannian Manifolds

%\section{Hodge Theory on Riemannian Manifolds}
%\label{sec:hodge_theory}

Hodge theory establishes fundamental connections between the topology, geometry, and analysis of smooth manifolds. We rapidly introduce elements of this theory  so as to fix our notations and conventions. In this section we work on a  connected, {\em oriented} Riemannian manifold $(H,h)$ of dimension $n$, to be later interpreted as the horizon.  All structures shall depend on the metric $h$ and orientation.  Hodge theory for $C^3$ metric can already be tricky so, as already mentioned, we shall assume that the metric is smooth to avoid dealing with complications related to regularity.

%This section presents key concepts for connected, oriented Riemannian manifolds $(M,g)$ of dimension $n$.

%%\subsection{Exterior Algebra Conventions}
The exterior algebra $\Lambda^* T^*H$ is embedded in the covariant tensor algebra via antisymmetrization in such a way that, for 1-forms $\alpha, \beta$, the wedge product satisfies \cite{misner73,spivak79,lee13} (Spivak's convention) $\alpha \wedge \beta = \alpha \otimes \beta - \beta \otimes \alpha$.
This extends to arbitrary $k$-forms through multilinear antisymmetric combinations: For 1-forms $\alpha_1, \dots, \alpha_k$,
\begin{equation}
\alpha_1 \wedge \cdots \wedge \alpha_k = \sum_{\sigma \in S_k} \text{sgn}(\sigma) \, \alpha_{\sigma(1)} \otimes \cdots \otimes \alpha_{\sigma(k)},
\end{equation}
where $S_k$ is the symmetric group. The exterior derivative $\dd : \Omega^k(H) \to \Omega^{k+1}(H)$ satisfies $\dd^2 = 0$ and the Leibniz rule
$\dd(\alpha \wedge \beta) = \dd\alpha \wedge \beta + (-1)^{\deg \alpha} \alpha \wedge \dd\beta$.

%\subsection{Metric Structures on Forms}
The metric induces a pointwise inner product on $k$-forms. For 1-forms:
\begin{equation}
\inner{\alpha}{\beta}_h = h^{ij} \alpha_i \beta_j \quad \text{(Einstein summation)}.
\end{equation}
For decomposable $k$-forms:
\begin{equation}
\inner{\alpha_1 \wedge \cdots \wedge \alpha_k}{\beta_1 \wedge \cdots \wedge \beta_k}_h = \det\bigl( \inner{\alpha_i}{\beta_j}_h \bigr).
\end{equation}
The global $L^2$ inner product for compactly supported forms is:
\begin{equation}
(\alpha, \eta) = \int_H \inner{\alpha}{\eta}_h  d h,
\end{equation}
where $d h$ is the Riemannian volume form, normalized such that $d h(e_1, \dots, e_n) = 1$ for any positively oriented orthonormal frame $\{e_i\}$.

The orientation and metric define the {\em Hodge star} $\hodgestar : \Omega^k(H) \to \Omega^{n-k}(H)$, characterized by:
\begin{equation}
\alpha \wedge \hodgestar \eta = \inner{\alpha}{\eta}_h  d h, \quad \forall \alpha, \eta \in \Omega^k(H).
\end{equation}
Equivalently, it can be characterized via $\star(e^1\wedge\cdots \wedge e^k )=e^{k+1}\wedge\cdots \wedge e^n$ for every $k$, where $\{e^i\}$ is any cobasis dual to a positively-oriented  basis $\{e_i\}$

Key properties include:
\begin{align}
&\hodgestar 1 = d h, \quad \hodgestar d h = 1, \\
&\hodgestar \hodgestar = (-1)^{k(n-k)} \quad \text{on} \quad \Omega^k(H),\\
&\hodgestar v^\flat=i_v d h \quad \text{for} \ v\in TH.
\end{align}
The $L^2$ inner product satisfies
\[
(\alpha, \eta)_{L^2} = \int_H \alpha \wedge \hodgestar \eta \quad \textrm{and} \quad (\hodgestar \alpha, \eta)_{L^2} = \int_H \alpha \wedge \eta.
\]
%\subsection{Codifferential and Laplace-Beltrami Operator}
The codifferential $\delta : \Omega^k(H) \to \Omega^{k-1}(H)$ is the formal adjoint of $\dd$:
\begin{equation}
(\dd \alpha, \eta) = (\alpha, \delta \eta), \quad \alpha \in \Omega^{k-1}(H), \eta \in \Omega^k(H),
\end{equation}
with $\delta \equiv 0$ on $\Omega^0(H)$. It relates to $\dd$ and $\hodgestar$ via \cite{spivak79,warner83,lee13}:
\begin{equation}
\delta = (-1)^{n(k+1) + 1} \hodgestar \dd \hodgestar.
\end{equation}
The Laplace-Beltrami operator $\laplace : \Omega^k(H) \to \Omega^k(H)$ is defined as:
\begin{equation}
\laplace = \dd \delta + \delta \dd.
\end{equation}
This operator is symmetric and non-negative with respect to $(\cdot, \cdot)$.

%\subsection{Harmonic Forms}
A $k$-form $\omega$ is \textit{harmonic} if $\laplace \omega = 0$. On a compact manifold $H$, this condition simplifies significantly \cite[Lemma 28]{petersen06}\cite{jost11}:
\begin{theorem} \label{jxdr}
For $\omega \in \Omega^k(H)$ on a connected oriented compact manifold,
\begin{equation}
\laplace \omega = 0 \iff \dd \omega = 0 \;\;\text{and}\;\; \delta \omega = 0.
\end{equation}
\end{theorem}
\begin{proof}
Using the $L^2$ inner product:
\[
(\laplace \omega, \omega) = (\dd \delta \omega, \omega) + (\delta \dd \omega, \omega) = \|\delta \omega\|_{L^2}^2 + \|\dd \omega\|_{L^2}^2.
\]
Since $(\cdot, \cdot)$ is positive definite, $\laplace \omega = 0$ implies $\dd \omega = 0$ and $\delta \omega = 0$. The converse is immediate. $\hfill$ \qed
\end{proof}
%\subsection{Hodge Decomposition Theorem}
The fundamental structural result for compact manifolds is \cite{spivak79,warner83}:
\begin{theorem}[Hodge Decomposition]
Let $(H,h)$ be a connected oriented compact Riemannian manifold. The space of $k$-forms decomposes orthogonally as:
\begin{equation}
\Omega^k(H) = \dd\Omega^{k-1}(H) \oplus \delta\Omega^{k+1}(H) \oplus \harm^k(H),
\end{equation}
where $\harm^k(H) = \{ \alpha \in \Omega^k(H) \mid \laplace \alpha = 0 \}$. This decomposition satisfies:
\begin{enumerate}
\item $\harm^k(H)$ is finite-dimensional
\item $\harm^k(H) \cong H^k_{\mathrm{dR}}(M)$ (de Rham cohomology)
\end{enumerate}
Thus, every de Rham cohomology class contains a unique harmonic representative.
\end{theorem}

This theorem establishes harmonic forms as geometric representatives of topological invariants.
%, demonstrating deep interplay between analysis, geometry, and topology.

%The proof of the following well known fact \cite[Lemma 3.8]{egidi23} uses the commutativity of the Lie derivative $L_K$ with $d$, and also with $\delta$ for $K$ Killing

The following fact is well known, e.g.\ \cite[Lemma 3.8]{egidi23}. We include the proof for completeness

\begin{lemma} \label{cotr}
On the compact connected orientable Riemannian space $(H,h)$ let $\alpha$ be a harmonic form and let $K$ be a Killing field, $L_Kh=0$, then $L_K\alpha=0$.
\end{lemma}

\begin{proof}
Indeed, $L_K\alpha$ is closed as $\dd L_K\alpha=L_K\dd \alpha=0$ (here we did not use the Killing property),  and also coclosed because $L_K$ commutes with $\delta$ (as $\delta$ depends just on the metric) and so $\delta L_K\alpha=L_K\delta \alpha=0$. As $L_K\alpha$ is both closed and coclosed, it is harmonic. But by Cartan's formula $L_K \alpha=\dd (i_K\alpha)$ which proves that $L_K\alpha$ is both harmonic and exact, thus, by Hodge decomposition, $L_K \alpha=0$. $\hfill$ \qed
\end{proof}

Hodge theory will just be used to prove the following key result in Riemannian geometry.

 \begin{theorem} \label{mxrt}
Let $(H,h)$ be a compact connected orientable Riemannian manifold, let $n$ be a Killing field and let $\omega$ be a smooth 1-form such that $i_n\dd \omega=0$. Then there is a smooth function $f$ and a constant $c$ such that $n(f)=\omega(n)-c$.
% suppose the dominant energy condition holds and let $H$ be a smooth compact degenerate horizon whose Riemannian flow is isometric. Then it admits a smooth tangent vector field of zero surface gravity (i.e.\ geodesic).
 \end{theorem}

 \begin{proof}
 %By Prop.\ \ref{nncd} it is sufficient to prove the statement for $H$ orientable.
 %Since the flow is isometric we can  choose a lightlike vector field $n$ which is Killing for some  Riemannian metric $h$ on $H$.
 %Let $\omega$ be the 1-form defined by Eq.\ (\ref{pprt}).
 From the Hodge decomposition, we have
 \begin{equation}
 \label{equ 1}
 \omega = \alpha + \beta + \dd f
 \end{equation}
  where $\alpha$ is a harmonic 1-form ($\Delta \alpha = 0)$; $\beta$ is a co-exact, $\beta  = \delta \eta$ for some 2-form $\eta$, and
 $f$ is a function on $H$.
 %Recall that $\Delta = d\circ\delta + \delta\circ d$ where $\delta$ is the co-differential operator.
 Since $\alpha$ is harmonic and $n$ is Killing we have by Lemma \ref{cotr}, Cartan's magic formula and Thm.\ \ref{jxdr},
 \[
 0=L_n\alpha = i_n \dd \alpha +\dd (\alpha(n))=\dd (\alpha(n)),
  \]
  which implies $\alpha (n) = a = \mbox{const}$.
 It is well-known that as $n$ is Killing, $L_n$ commutes with $\delta$ (as we recalled in the proof of Lemma \ref{cotr}), so $L_n \beta=L_n\delta \eta=\delta L_n \eta$ which shows that $L_n \beta$ is co-exact.
 %\[
 %\delta (L_n\beta) = L_n(\delta \beta) = 0.
 %\]
 From (\ref{equ 1}), we have $\dd\omega = \dd\beta$ since $\dd\alpha = 0$ (a harmonic form is both closed and co-closed). Then we get $i_n \dd\beta = i_n \dd\omega = 0$.  Using Cartan's magic formula, it follows that $L_n\beta = \dd(i_n(\beta)) = \dd(\beta(n))$. We have proved that $L_n\beta$ is both co-exact and exact, which by Hodge $L_2$-orthogonal  decomposition implies that $L_n \beta$ vanishes.
 So, we  proved that $0=L_n\beta = \dd(\beta(n))$, which means that
 $\beta (n) = b = \mbox{const}$.
Evaluating (\ref{equ 1}) over $n$ we get  $\kappa = c + n(f)$ where $c:=a+b$ is a constant. $\hfill$ \qed
% Setting  $n^\prime = e^{-f}n$ we get from Eq.\ (\ref{ren2}) that   $\kappa^\prime = ce^{-f}$.
 % In the degenerate case, since by Thm.\ \ref{Birk app}, $\int_H \kappa^\prime d\tilde h= 0$, it follows that the constant $c$ is zero and hence $\kappa^\prime = 0$.
 \end{proof}

\begin{proposition} \label{nncd}
If for a given spacetime dimension and under the null energy condition orientable degenerate compact horizons admit zero surface gravity, then the same is true in the non-orientable case.

If for a given spacetime dimension and under the null energy condition orientable degenerate compact horizons with isometric flow admit zero surface gravity, then the same is true in the non-orientable case.
%A version where ``degenerate" is replaced by ``degenerate with isometric flow", also holds.
\end{proposition}

\begin{proof}
Let $\pi: \tilde M\to M$ be the (time-oriented) orientable double covering of $(M,g)$, then $\tilde H:=\pi^{-1}(H)$ is an orientable compact double covering of $H$ (which is disconnected if $H$ is already orientable). Every generator of $\tilde H$ projects on some complete generator and so is itself complete. Thus $\tilde H$ is degenerate. If $H$ is isometric for some choice of vector field and metric $(n,h)$, $L_nh=0$, then the same is true for $\tilde H$, it is sufficient to lift $n$ and $h$ by using the local diffeomorphism.

By assumption there is a lightlike field $\tilde n$ and corresponding 1-form $\tilde \omega$ on $\tilde H$ such that $\nabla \tilde n= \tilde n \otimes  \tilde \omega$, $\tilde \omega(\tilde n)=0$.

Every point $q$ on $H$ has an open neighborhood $U$ which is (spacetime) isometric to open neighborhoods $\tilde U_1$ and $\tilde U_2$ of the inverse images $\tilde q_1$ and $\tilde q_2$ of  $q$ (the actual labelling is irrelevant as it will be used locally to define objects in $U$ which are well defined as invariant under exchange $1 \leftrightarrow 2$).

For each $p\in U$ we have inverse images $\tilde p_1\in \tilde U_1$ and $\tilde p_2\in \tilde U_2$ of the point $p$.
% Every point $p$ has an open neighborhood $U$ which is (spacetime) isometric to open neighborhoods $\tilde U_1$ and $\tilde U_2$ of the inverse images $\tilde p_1$ and $\tilde p_2$ of $p$ (the actual ordering is irrelevant as it will be used locally to define objects in $U$ which are well defined as invariant under switching $1 \leftrightarrow 2$).
 Let  $n_i=\pi_*(\tilde n(\tilde p_i))$ and similarly project $\omega$ with the diffeomorphism.
 Then with obvious notation and identifications we have on $p$, and hence on $U$, $\nabla n_i= n_i \otimes\omega_i$ and $\omega_i(n_i)=0$. Let $n=n_1+n_2$, and observe that, once $\tilde n$ has been chosen, the construction of this vector field can be repeated for a covering $\{U_i\}$ of $H$ so leading to a well defined  future-directed lightlike field over $H$ (again because of invariance under exchange $1 \leftrightarrow 2$). In particular, there is $\omega: H\to T^*H$ such that $\nabla n=n \otimes \omega$.

  On each element of the covering we have $\nabla n= n_1 \otimes  \omega_1+ n_2  \otimes\omega_2$ but also $\nabla n=  n \otimes \omega$. Taking the trace of the former equation $\textrm{tr}(\nabla n)=\omega_1(n_1)+\omega_2(n_2)=0$ where we used vanishing of surface gravity on $\tilde H$, thus $0=\textrm{tr}(\nabla n)=\omega(n)$, which means vanishing of surface gravity with the choice $n=n_1+n_2$ all over $H$. $\hfill$ \qed
\end{proof}

We are ready to prove the main result of this section. It has no restrictions on the spacetime dimension.

 \begin{theorem} \label{mcrt}
On a smooth spacetime $(M,g)$ suppose the null energy condition holds and let $H$ be a smooth compact degenerate horizon whose Riemannian flow is isometric. Then it admits a smooth tangent vector field of zero surface gravity (i.e.\ geodesic).
 \end{theorem}

 \begin{proof}
 By Prop.\ \ref{nncd} it is sufficient to prove the statement for $H$ orientable.
 Since the flow is isometric we can  choose a lightlike vector field $n$ which is Killing for some  Riemannian metric $h$ on $H$.
 Let $\omega$ be the 1-form defined by Eq.\ (\ref{pprt}).  From the null energy condition and Cor.\ \ref{cooe} we have $i_n\dd \omega=0$. As all the assumptions of Theorem \ref{mxrt} are satisfied there are a smooth function $f$ and a constant $c$ such that $n(f)=\omega(n)-c=\kappa-c$.
Setting  $n^\prime = e^{-f}n$ we get from Eq.\ (\ref{ren2}) that   $\kappa^\prime = ce^{-f}$ which has the sign of $c$.

  In the degenerate case, since by Thm.\ \ref{Birk app}, $\int_H \kappa^\prime d\tilde h'= 0$ (here $h'=g_T+\alpha'\otimes \alpha'$, $\alpha'=e^f \alpha$) it follows that the constant $c$ is zero and hence $\kappa^\prime = 0$. $\hfill$ \qed
 \end{proof}

The following proof could be completely skipped as we shall give two more  proofs independent of dimension in the following sections.

\begin{theorem} \label{mcrf}
On a smooth   spacetime $(M,g)$ of dimension 2,3 or 4, suppose the null energy condition holds and let $H$ be a smooth compact degenerate horizon. Then it admits a smooth lightlike tangent vector field $n$ of zero surface gravity (i.e.\ geodesic).
\end{theorem}

\begin{proof}
By Prop. \ref{nncd} it is sufficient to consider the orientable case, and by Thm.\ \ref{mcrt} it is sufficient to consider the non-isometric case. In dimension 2 and 3 this concludes the proof as there are no orientable compact horizons with non-isometric Riemannian flow, see Thm.\ \ref{carf}.
In dimension 4 we are left with the  structure of case (iii-b) of Theorem \ref{car}. This is explored in  Carri\`ere original paper \cite{carriere84} and \cite[Example 2]{minguzzi24c}. It is known as the hyperbolic torus $T^3_A$, where $A\in SL(2, \mathbb{Z})$, $\textrm{tr} A>2$, because it can be obtained as a suspension of base $T^2$ with the resulting flow being Anosov's (roughly ``suspension'' here refers to the construction which goes as follows:  consider the product $T^2\times [0,1]$ and identify the two $T^2$ bases after acting with $A$ on one of them). However, in Carri\`ere non-isometric construction the Riemannian flow is not the natural one of the suspension but that of one of the two  eigenvectors of $A$.

We continue with two independent arguments. The former uses the possibility of solving the cohomological equation via Fourier analysis over torus bundles, under a Diophantine condition, as shown by Dehghan-Nezhad and El Kacimi Alaoui \cite{dehghan07} (our argument moves from their results so we do not need further Fourier type arguments). The latter makes use of some deep and powerful  results by {{\'A}lvarez L{\'o}pez} and Dom{\'\i}nguez on the possibility of finding a bundle like metric with a basic mean curvature form  $\mu$ (we recalled this concept in our previous work \cite{minguzzi24c}).
%In the next sections we shall really obtain further proofs as the 4-dimensional case will follow from the general case.

\begin{itemize}
\item[-]{\em Proof via cohomological equation for $T^2$ bundles.} \\
 Topologically, $T^3_A$ is  a $T^2$-bundle over $S^1$, $\pi:T^3_A\to S^1$.  We denote with $u$ the standard coordinate of $S^1$ identified as $[0,1]$. For any interval $I\subset S^1$, $\pi^{-1}(I)$ is diffeomorphic to $I\times T^2$ and the flow is conjugate (without parameter), by the same diffeomorphism, to the linear flow of a constant eigenvector of $A$ (that is, consider the corresponding linear flow on $T^2$ and lift it to $I\times T^2$.). This eigenvector has rationally independent components and so the tori fibers are really the closures of the orbits of Carri\`ere non-isometric flow.

 We have shown in \cite[Example 2]{minguzzi24c} in some detail, that the two roots $\lambda_+,\lambda_-$ of the equation $\lambda^2-\textrm{tr} (A) \lambda +1$  are irrational. We set $v_+=(a,b)^T$ for one eigenvector of $A$ and by imposing $A(v_+)=\lambda_+ v_+$, we find $a,b\ne 0$ with an irrational ratio $c=b/a$. A similar conclusion is reached considering $\lambda_-$ and $v_-$. As any eigenvector can be rescaled, we can consider $v_+=(1,c)$ where the first line of $A(v_+)=\lambda_+ v_+$  gives $A_{11}+ c A_{12}=\lambda_+$. But $A$ has integer coefficients so $A_{12}\ne 0$ and hence $c=\frac{\lambda_+-A_{11}}{A_{12}}$. As all the numbers on the right-hand side are algebraic so is $c$. Roth's theorem \cite{roth55,hindry00} states that for every   irrational algebraic number $c$ there are positive constants $\tau>1$ and $C<1$
such that for every integers $p,q$, $q\ne 0$,   $\vert c+p/q\vert \ge C/\vert q\vert ^{1+\tau}$. As a consequence, for every   irrational algebraic number $c$ there are positive constants $\tau>1$ and $C<1$ such that for every integers $p,q$, $\vert p+q c \vert \ge  C/\max(\vert p\vert,\vert q\vert) ^{\tau}$, which  implies that  for every vector $z\in \mathbb{Z}^2$, $ \vert z \cdot v_+ \vert \ge   C/\max(\vert z_1\vert,\vert z_2\vert) ^{\tau}$. In other words, $v_+$ is a Diophantine vector \cite{schmidt80,dehghan07}. An analogous conclusion holds for $v_-$.

By Remark \ref{ccor} it is sufficient to prove that the cohomological equation can be solved over $T^3_A$ endowed with Carri\`ere's non-isometric flow. Let $n$ be the vector field generating the flow as provided in coordinates, for instance, in \cite[Ex.\ 2]{minguzzi24c} (there it is denoted $Y$ and $u$ is denoted $x$). There it is shown that in the local trivialization it is the constant  vector $v_{\pm}$ up to a factor dependent on $u$.

 If we can solve the cohomological equation on sets of the form $\pi^{-1}(I)$ then we can solve it all over the bundle, it is sufficient to lift a partition of unity relative to a covering $\{I_i\}$ of   $S^1$ to get a partition of unity $\{\varphi_i\}$ over the bundle relative to a covering $\pi^{-1}(I_i)$, then solve $n(f_i)= \kappa$ on $\pi^{-1}(I_i)$, multiply by $\varphi_i$, sum over $i$, and use $n(\varphi_i)=0$ to get that $f=\sum_i \varphi_i f_i$ provides a solution. We can thus restrict ourselves to a set of the form $\pi^{-1}(I)$ and so work on $I \times T^2$ with a vector field $n=r(u) n'$ where $n'$ is a constant Diophantine vector. As
 %We proved (this step uses the dominant energy condition) a vanishing integral property $\int_{T^2} \kappa \dd \bar g=0$ in \cite[Prop.\ 24]{minguzzi24c} on each fiber $T^2$.
 %But $n=r(u) n'$ where $n'$ is a constant Diophantine vector, as
 $n(u)=0$, from Eq.\ (\ref{ren2}), $\kappa=r(u) \kappa'$.

 We proved in Theorem \ref{Birk app}, see also \cite[Prop.\ 24]{minguzzi24c}, (this step uses the null energy condition) the vanishing integral property over each $T^2$-fiber (hence constant $u$), $\int_{T^2} \kappa \dd \bar g=0$. Here $\dd \bar g$   is a measure on each fiber torus which is invariant under the flow of $n$, and derived from a particular Riemannian metric. It is also invariant  under the flow of $n'$ because $n(u)=0$ (use the expression of the divergence of $r(u)n'$). This implies that $\int_{T^2} \kappa' \dd \bar g=0$ where   $\dd \bar g$ is invariant under the flow of $n'$.
However, the Diophantine vector $n'$ has rationally independent components, and for such a linear flow we have unique ergodicity\footnote{In the map case see, for instance, \cite[Cor.\ 4.15]{einsiedler11}. For the flow case that interests us see, e.g.\, J.-F.\ Quint, {\em Examples of unique ergodicity of algebraic
flows}, Lectures at Tsinghua University, Beijing, November 2007, Example 1.2.11.
 The flow case can be proved imposing for any continuous function $f: T^n \to \mathbb{R}$ and any $t \in \mathbb{R}$,
$\int_{T^n} f(\phi_t(x)) \, d\mu(x) = \int_{T^n} f(x)  d\mu(x)$, and proceeding
with a Fourier expansion of $f$. For related material see \cite[Sec.\ 51]{arnold89}.} that is, there is really only one invariant Borel probability measure on each fiber torus of $I\times T^2$, and  in place of $\dd \bar g$ we could use the canonical one coming from the trivialization.
%for the fibers in $\pi^{-1}(I)$ (and hence, by covering the bundle, for every fiber)

% Now,  note that if we prove that the cohomological equation can be solved for any $\kappa$ for some $n$ then it can be solved for every $n'$ obtained rescaling $n$.  By using the diffeomorphism we can work on $I \times T^2$. Carri\`ere construction shows that the flow on $\pi^{-1}(I)$ is conjugate, without parameter, to a linear flow on $I \times T^2$ induced by a constant vector. We can now rescale $n$ and consider the cohomological equation for that constant vector field.

%Observe that We are left with the problem of establishing if the cohomological equation can be solved in a manifold of the form $\pi^{-1}(I)$, and hence, by the diffeomorphism, if we can do it over $I\times T^2$ with the flow described as above (note that if ).

This geometry, in which we have a torus bundle over a manifold (connected and orientable but not necessarily compact), and a Riemannian flow which can be represented, in a local trivialization of the bundle, by a Diophantine linear flow on the torus fiber independent of the coordinate(s) $u$ of the base, has been considered by  Dehghan-Nezhad and  El Kacimi Alaoui \cite{dehghan07}. Their result \cite[Thm.\ 2.4, also p.\ 1114-15]{dehghan07} states that if $n'$ is a Diophantine vector field generating the Riemannian flow, then the cohomological equation $n'(f)=\kappa'$ has a smooth solution for any smooth $\kappa'$ such that the integral of $\kappa'$ on every torus fiber vanishes. Here the measure to be used is the canonical one of the torus in the trivialization. As explained in the previous paragraphs, we just need to apply their theorem to $I\times T^2$, and the same paragraphs show that the integral assumption is satisfied. By multiplying $n'(f)=\kappa'$ by $r(u)$ we get that the solution $f$ also solves $n(f)=\kappa$, and now we can use the partition of unity to globalize.

%Note that the vector field (constant in the trivialization) can be extended from $\pi^{-1}(I)$ to the whole bundle. Let us call it $n$.
%We proved (this step uses the dominant energy condition) a vanishing integral property $\int_{T^2} \kappa \dd \bar g=0$ in \cite[Prop.\ 24]{minguzzi24c} but for a measure $\dd \bar g$ which is invariant on each fiber torus under the flow of $n$, and  which is derived from a certain Riemannian metric.

%However, a Diophantine vector has rationally independent components, and for such a linear flow we have unique ergodicity\footnote{In the map case see, for instance,  M. Einsiedler and T. Ward, {\em Ergodic Theory with a view towards Number Theory} (Springer-Verlag, London, 2011), Cor.\ 4.15. For the flow case that interests us see, e.g.\, J.-F.\ Quint, {\em Examples of unique ergodicity of algebraic
%flows}, Lectures at Tsinghua University, Beijing, November 2007, Example 1.2.11.
 %The flow case can be proved imposing for any continuous function $f: T^n \to \mathbb{R}$ and any $t \in \mathbb{R}$,
%$\int_{T^n} f(\phi_t(x)) \, d\mu(x) = \int_{T^n} f(x)  d\mu(x)$, and proceeding
%with a Fourier expansion of $f$. For related material see V. I. Arnold {\em Mathematical Methods of Classical Mechanics}    (Springer-Verlag, New York, 1989), Sec.\ 51.} that is, there is really only one invariant measure on each fiber torus of $I\times T^2$, and so
%for the fibers in $\pi^{-1}(I)$ (and hence, by covering the bundle, for every fiber)
%we conclude that the integral condition of \cite[Thm.\ 2.4]{dehghan07} is satisfied.
%All assumptions are satisfies thus the
In conclusion, the cohomological equation $n(f)=\kappa$ over $T^3_A$
can be solved for $\kappa=\omega(n)$,  and so the horizon admits a smooth lightlike tangent field of zero surface gravity cf.\ Remark \ref{ccor}.

\item[-]{\em Proof via {{\'A}lvarez L{\'o}pez}'s and Dom\'{\i}nguez's theorems.}\\
We recall that if $g_B=g_T+\chi \otimes \chi$, is the bundle-like metric (i.e.\ $L_X g_B=0$ for one and every vector field tangent to the foliation and $\chi$ is positive on the foliation, see e.g.\ \cite[Prop.\ 2]{minguzzi24c}), then $n$ is defined by the normalization $\chi(n)=1$ and the mean curvature 1-form is $\mu:=i_n\dd \chi$. It is basic if $i_n \mu =0$, which is always the case, and $i_n \dd \mu=0$ which does not necessarily hold. {{\'A}lvarez L{\'o}pez} has shown that the basic component $\mu_b$ (using orthogonality in the sense of forms induced by $g_B$) determines a class (the {\'A}lvarez  class) $[\mu_b]$ which does not depend on the bundle-like metric \cite{alvarez92}. In any case, by Dominguez's theorem \cite{dominguez95,dominguez98}\cite[p.\ 81]{tondeur97}, it is possible to choose $\chi$ (and hence $g_B$) while leaving $g_T$ unaltered,  in such a way that the mean curvature 1-form $\mu$ is basic and closed. We define $n$ to be that positive vector such that $\chi(n)=1$ (equivalently, $g_B(n,n)=1$), for Dominguez's choice of $g_B$ and $\chi$. Observe that for us it is important that $g_T$ does not change \cite[Cor.\ 4.23]{dominguez98}, as it is the metric induced on the horizon. (Our proof below works using just {{\'A}lvarez L{\'o}pez}'s results \cite{alvarez92}, it is sufficient to work with $\mu_b$ instead of $\mu$.)

 From the null energy condition, $d\omega(n, \cdot) = 0$ and $L_nd\omega = 0$ so $d\omega$ is a basic 2-form which is closed so it defines an element in $H_b^2(\mathcal{F})$. By \cite[Prop.\ III.B.2]{carriere84},   we have for the hyperbolic torus $T^3_A$ endowed with Carri\`ere non-isometric flow, $H_b^2(\mathcal{F}) = 0$. So there exists a basic 1-form $\alpha$ such that  $d\omega = d\alpha$. This means that $\omega - \alpha$ is a closed 1-form. But the hyperbolic torus has Betti number equal to $1$, this is a well known result,  see \cite{funar13} (paragraph before Theorem 1.1 and the last paragraph of the proof of Proposition 7.2), hence $\textrm{dim} H^1(T^3_A)=1$.
  %and \cite{dathe08}  (case 1 in the proof of Theorem 4.1).
 Since the flow is not isometric, $[\mu](=[\mu_b])$ defines a non zero class in $H^1_b(H)$, see \cite[Thm.\ 6.4]{alvarez92}\cite[Sec.\ 2.2]{minguzzi24c}, and hence  in $H^1(H)$ (Remark \ref{crre}). Since the dimension of $H^1(H)$ is $1$, there exists a constant $\lambda$ and a function $f$ such that
 $\omega - \alpha = \lambda\mu + df$. Since $\mu$ and $\alpha$ are basic 1-forms, we have $\alpha (n) = \mu (n) = 0.$ So $\kappa = \omega(n) = n(f).$
 Setting $n^\prime = e^{-f}n$ we get
  $\kappa^\prime = 0$. $\hfill$ \qed
\end{itemize}
\end{proof}

\section{Carri\`ere's neighborhoods and generic dimension} \label{nnnt}

According to Proposition \ref{nncd} we can focus on compact connected {\em orientable} manifolds endowed with a Riemannian flow.  Carri\`ere \cite{carriere84} for this type of structure proved that
%the closure of any orbit is a torus $T^{k+1}$.  The neighborhood of such torus has a particular structure established in \cite[Prop.\ 4]{carriere84}.
the closures of orbits of a Riemannian flow \(\phi\) on a compact  connected orientable manifold are tori, and that restricted to each torus, the flow \(\phi\) is conjugate (without parameter) to a linear flow. He then established the structure of \(\phi\) in a neighborhood of the closure of an orbit \cite[Prop.\ 4]{carriere84}. Due to their useful structure, these neighborhoods are now called {\em Carri\`ere's neighborhoods} \cite{royo01,nozawa14}.

Let \(D^p\) denote the unit Euclidean open ball centered at the origin in \(\mathbb{R}^p\). Let \(\phi\) be a Riemannian flow on a compact connected orientable manifold \(H^{n}\) and \(L\) the closure of an orbit of \(\phi\). The manifold \(L\) is diffeomorphic to \(T^{k+1}\), $k\ge 0$, and

\begin{proposition}[Existence of Carri\`ere's neighborhoods] \label{bbe}
There exists a neighborhood \(V\) of \(L\) saturated by \(\phi\) such that:
\begin{enumerate}
    \item[(i)] \(V\) is diffeomorphic to \(S^1 \times T^k \times D^{n-k-1}\) via a diffeomorphism sending \(L\) to \(S^1 \times T^k \times \{0\}\).
    \item[(ii)] The flow \(\phi\) restricted to \(V\) is conjugate (without parameter) to the flow obtained by suspending a diffeomorphism \(\gamma\) of \(T^k \times D^{n-k-1}\) of the form \(\gamma(x,y) = (B(x), O(y))\), where \(B\) is an irrational translation of \(T^k\) and \(O\) a rotation of \(\mathbb{R}^{n-k-1}\).
\end{enumerate}
\end{proposition}
That is, $V$ is diffeomorphic to $[0,1]\times T^k \times D^{n-k-1} / \sim$
\[
(1,x,y) \sim  (0,x+v, O(y))
\]
where $(1,v_1,\cdots, v_k)$ has rationally independent components and $O$ has unit determinant (due to the orientability of $H$ and hence $V$). The flow is given by
\[
s\mapsto (s, x_0, y_0)
\]
and the interesting dynamics comes from the identification.

\begin{remark} \label{nntr}
There is a smooth function $\rho \ge 0$ on $H$ which is positive on $L$, is constant on every orbit (and hence on their closures), and vanishes outside the Carri\`ere's neighborhood $V$. Indeed, by using point (i) it is sufficient to consider a smooth non-negative function $\rho$ on $H$ that is zero outside $V$ and on $V$ is the lift of a function on $D^{n-k-1}$ that depends only on the radial coordinate $r:=\Vert y\Vert$, is equal to one  for $r<1/2$  and zero for $r>1-\epsilon$, $\epsilon<1/2$. As a consequence, by covering $H$ with  the open sets $\{\rho >0\}$ for each $\rho$ so constructed, and extracting a finite subcovering we  get a finite covering of Carri\`ere's neighborhoods  $\{V_i\}$ and the existence of an associated partition of unity $\{\varphi_i\}$,  $\varphi_i=\rho_i/\sum_j \rho_j$. Note that the functions $\varphi_i$ are constant over the orbits as so are the $\rho_i$.
\end{remark}
%\begin{proposition}
%If for every Carri\`ere's neighborhood $V$ there exists  a vector field with zero surface gravity, then there also exists a global field on $H$ with zero surface gravity.
%\end{proposition}

\begin{proposition}
Let $\kappa$ be smooth function on $H$ and suppose that
%If there is a constant $c$ such that
for every Carri\`ere's neighborhood $V$ there exists  a vector field with  surface gravity equal to  $\kappa$ on $V$, then there also exists a global field on $H$ of surface gravity equal to $\kappa$.
\end{proposition}

We shall only be interested in the case $\kappa=0$ (zero surface gravity).
%As the proof shows, we could also replace the property of being constant for surface gravity with the property of being negative or positive.

\begin{proof}
Indeed, if such $n_i$ exists for each element of the covering $V_i$, then $\nabla n_i=n_i\otimes \omega_i$ for some 1-form $\omega_i: V_i \to T^*V_i$, $\kappa_i=\textrm{tr}(n_i\otimes \omega_i)=\omega_i(n_i)=\kappa$, but then, defining $n:=\sum_i \varphi_i n_i$
\[
\nabla n=\sum_i \varphi_i n_i \otimes[\omega_i + \dd \log \varphi_i]
\]
from which we get
\[
\kappa=\textrm{tr}(\nabla n)=\sum_i \varphi_i [\kappa_i+ \p_{n_i}(\log \varphi_i)]=\sum_i \varphi_i \kappa=\kappa
\]
where we used the fact that the functions of the partition of unity $\varphi_i$ are constant over the orbits. $\hfill$ \qed
\end{proof}

The same result can be expressed in terms of solutions to the cohomological equation.

\begin{proposition} \label{mots}
Let $H$ be a compact  connected orientable manifold endowed with a Riemannian flow generated by a smooth vector field $n$. Let $\kappa$ be a smooth function on $H$. If the cohomological equation $n(f)=\kappa$ admits solutions over every Carri\`ere's neighborhood $V$ then it also admits a global solution.
\end{proposition}

Again, we shall be mostly interested in the case $\kappa=0$.

\begin{proof}
Indeed, let us choose a finite covering $\{V_i\}$. From $n(f_i)=\kappa$ we have multiplying by $\varphi_i$, using $n(\varphi_i)=0$ and summing, $n(f)=\kappa$, where $f=\sum_i\varphi_i f_i$. $\hfill$ \qed
\end{proof}

\begin{theorem}
\label{Birt}
%Suppose that the spacetime satisfies the dominant energy condition and let $n$ be any smooth lightlike field tangent to a compact degenerate horizon $H$.
Let $H$ be a compact connected orientable manifold endowed with a Riemannian flow generated by a smooth vector field $n$, and let $\omega$ be a 1-form on $H$ which satisfies $i_n\dd \omega=0$. Suppose that the equivalent properties  of Theorem \ref{Birkm} are satisfied.
Then there is a smooth function $f$ such that $n(f)=\omega(n)$, and hence the 1-form $\omega-\dd f$ is basic.
\end{theorem}

%As the notation suggests, there is no need to specify the measure used in the integral condition. By Carri\`ere result the closure of the orbit is diffeomorphic to a torus and the flow is conjugate to a linear flow on the torus with dense orbits, hence ergodic. By the measure we intend the canonical one of the torus, or that induced from a bundle-like metric. As argued they differ by a constant factor.

%\begin{remark}
%As pointed out in Remark \ref{cnnqa} the assumptions in the theorem statement are independent of $n$. The conclusion is also independent of it.
%If the equation $n(f)=\omega(n)$ can be solved for some vector field $n$ then the same holds for any choice of vector field. Indeed, if $n'=e^{-g} n$ then by Eq.\ (\ref{ren1}), $\omega'(n')=(\omega-\dd g)(n')=e^{-g} n(f)- n'(g)=n'(f-g)$, thus $f'=f-g$.
%\end{remark}

\begin{proof}
%By the linearity in $n$ (for fixed $\omega$), if we can solve $n(f)=\omega(n)$ for some $n$, we can solve it for every $n$.
By Prop.\ \ref{mots} for $\kappa=\omega(n)$, it is sufficient to solve the same problem on a Carri\`ere's neighborhood $V$. We shall omit the diffeomorphism from $V$ to $[0,1]\times T^k\times D^{n-k-1}/\!\sim$, understanding that $\omega$ and $n$ can be transferred with the diffeomorphism if needed. At no step of the argument we shall need the metric $g_T$ (which might be non-unique).

We denote with $n$ the original vector field and with $n'=\frac{\p}{\p s}$ the natural vector field of the suspension where $s$ is the parameter of $[0,1]$. This vector field is defined just on $V$.
However, its restriction to $V_{1/2}:=[0,1]\times T^k\times D_{1/2}^{n-k-1}/\sim$, where $D_{1/2}^{n-k-1}$ is the open ball of radius 1/2 of $\mathbb{R}^{n-k-1}$, can be easily extended to the whole $H$ (e.g.\ by taking a convex combination with another global choice, making use of function $\rho$, c.f.\ Remark \ref{nntr}). We denote this extension with $n''$. We recall that the properties of Theorem \ref{Birkm} do not depend on the vector field (Remark \ref{cnnqa}), thus they hold also using $n''$ which implies $\lim_{s\to \infty}\frac{1}{s}\int_0^s \omega(\dot x(s)) \dd s=0$, where $x$ is any orbit on $V_{1/2}$ parametrized with the natural suspension parameter $s$.

 Let $L$ be the closure of orbit corresponding to $r=0$. Let $\rho$ be the function on the ball $D^{n-k-1}$ constructed in Remark \ref{nntr}. We already observed that the region $r<1/2$, $r:=\Vert y \Vert$, provides another Carri\`ere's neighborhood of $L$ denoted $V_{1/2}$ as it is the suspension with base $T^k\times D_{1/2}^{n-k-1}$. The 1-form $\omega':=\rho\omega$ vanishes for $r>1-\epsilon$ and is coincident with $\omega$ for $r<1/2$ hence in $V_{1/2}$. Moreover, since $\rho$ is constant over every orbit in $V$, $i_n\dd \omega'=(i_n\dd \rho) \omega+\rho i_n\dd \omega=0$.
 %Finally, $\rho$ is constant over every orbit, and hence over their closure, which implies for every torus $T^a\subset V$, $\int_{T^a} \omega'(n)=\rho \int_{T^a} \omega(n)= 0$.

We are going to extend the manifold $[0,1]_s\times T^k\times D^{n-k-1}/\sim$ to a compact manifold via a compactification of the factor $D^{n-k-1}$. This will produce a new suspension and hence an extension of the Riemannian flow to the compactified space. We proceed as follows. We regard $D^{n-k-1}$ as a subset of $\mathbb{R}^{n-k-1}$ with the metric given by $\Omega(r)^2 g_{\mathbb{E}}$ where $\Omega(r)= \frac{2}{1+r^2}$, $r=\Vert y\Vert$, for $r\ge 2$ and $\Omega=1$ for $r \le 1$. The factor is that appearing in the stereographic projection. As a result we can add a north pole so embedding $D^{n-k-1}$ into a sphere $S^{n-k-1}_{\textrm{def}}$ where ``def'' reminds us that the metric has been deformed near the south pole to be flat, that is, the south pole has a ball neighborhood isometric to $D^{n-k-1}$. Now we consider a suspension $N= [0,1]\times T^k \times S^{n-k-1}_{\textrm{def}}/\sim$
with base $T^k \times S^{n-k-1}_{\textrm{def}}$. The map of the identification is best expressed using the coordinate of the Euclidean space, that is, using the stereographic projection, as follows (the north pole is sent to itself)
\[
(1,x,y) \sim  (0,x+v, O(y)).
\]
Formally, the map does not change, and it is still an isometry because the rotation preserves the function $\Omega$. This isometric property implies that it induces a Riemannian flow with tangent field $\p/\p s$ which we keep denoting $n'$. In conclusion, the Carri\`ere neighborhood can be regarded as a subset of a compact connected  manifold $N$ which is obtained via a  suspension. Coming from a suspension, the flow is isometric, in fact  the direct sum of metrics $h=\dd t^2\oplus g_{\mathbb{T}^k}\oplus g_{S^{n-k-1}_{\textrm{def}}}$ is such that $L_n h=0$ (see the discussion in \cite{minguzzi24c}). Furthermore, the orientability of $T^k \times S^{n-k-1}_{\textrm{def}}$ and the fact that $O$ has unit determinant implies that $N$ is orientable. Now, we extend $\omega$ further, defining $\omega'=0$ in $N\backslash V$. Then, $N$, $n'$, $h$, $\omega'$, satisfy the assumptions of Theorem \ref{mxrt} thus there are a smooth function $f: N \to \mathbb{R}$ and a constant $c$ such that $n'(f)=\omega'(n')-c$. We denote with $\check f$ the restriction to $V_{1/2}$. As there  $\omega'=\omega$ we get  $n'(\check f)=\omega(n')-c$. Integrating this equation and using $\lim_{s\to \infty}\frac{1}{s}\int_0^s \omega(\dot x(s)) \dd s=0$ and $n'=\frac{\dd }{\dd s}$ we get $c=0$ (because $f$,  and hence $\check f$, is bounded as it is a continuous function over the compact set $N$). Using the linearity with respect to $n$ we arrive at $n(\check f)=\omega(n)$ on $V_{1/2}$. As the equation can be solved over Carri\`ere neighborhoods that cover the whole $H$, by passing to a finite covering we get a global solution on $H$. $\hfill$ \qed
\end{proof}

We are ready to prove the main theorem

\begin{proof}[{\bf Proof of Theorem \ref{jdot}}]
Let  $(M,g)$ be a smooth spacetime (of any dimension) and assume the null energy condition. Let $H$ be a  compact degenerate horizon, we want to prove that there exists  a smooth future-directed lightlike tangent vector field $n$ of zero surface gravity, $\omega(n)=0$, where $\nabla_X n=\omega(X) n$.

By the null energy condition $i_n \dd \omega=0$ (Cor.\ \ref{cooe}).
By Prop.\ \ref{nncd} we can assume, without loss of generality, that $H$ is orientable. By Theorem \ref{Birk app} the equivalent properties of Theorem \ref{Birkm} are satisfied. By Theorem \ref{Birt} there is a smooth function such that $n(f)=\omega(n)$, thus by Eq.\ (\ref{ren2}) $n'=e^{-f}n$ has zero surface gravity. $\hfill$ \qed
\end{proof}
%
%
%In order to establish  the existence of the  vector field with zero surface gravity  on $V$ we just need to fix a natural geometric choice for $n$ on $V$, define $\kappa$ on $V$ via $\nabla_n n=\kappa n$,  and prove that the cohomological equation $n(f)=\kappa$ can be solved in $V$.
%
%It is crucially important for us to observe that, since $H$ is orientable, the rotation $O$ in Prop.\ \ref{bbe} must have unit determinant, i.e.\ it is a  special orthogonal transformation. The canonical form of a special orthogonal transformation in any dimension can be described as a block diagonal matrix where each block is either a $2\times 2$ rotation matrix or a diagonal matrix with entries +1. As a consequence, every special orthogonal matrix $O$ is exponential, namely we can find an antisymmetric matrix $A$ such that $O=\exp A$.
%

\section{Cohomological proof in generic dimension} \label{mnlc}

We provide an entirely cohomological proof of the main Theorem \ref{jdot}, building on several sophisticated results from Riemannian foliation theory. It is worth emphasizing that while in the previous section we handled the possibly non-isometric case through a relatively simple geometric covering argument, the results we now employ draw upon the full machinery of harmonic function theory — although this technical depth remains largely concealed behind the cohomological statements we invoke (see, for instance, the proof of \cite[Lemma 3.4]{schliebner15}).

Our argument relies on Royo Prieto's generalization to Riemannian flows of the Gysin sequences relating de Rham and basic cohomology \cite[Thm. 3.2]{royo01b} (see also \cite[Thm. 1.42]{schliebner15}, \cite[Sec. 3]{larz10}), together with several other cohomological facts.

\begin{theorem}
\label{tirt}
%Suppose that the spacetime satisfies the dominant energy condition and let $n$ be any smooth lightlike field tangent to a compact degenerate horizon $H$.
Let $H$ be a compact connected orientable manifold endowed with a Riemannian flow generated by a smooth vector field $n$, and let $\omega$ be a 1-form on $H$ which satisfies $i_n\dd \omega=0$. If the flow is not isometric then  there is a smooth function $f$ such that $n(f)=\omega(n)$, and hence the 1-form $\omega-\dd f$ is basic.
\end{theorem}

\begin{proof}
By a result due to Molino and Sergiescu  \cite{molino85} for every non-isometric Riemannian flows $H^{n-1}_b(H)=0$.
By \cite[Lemma 3.4]{schliebner15} we have for every non-isometric Riemannian flows $H^1_{dR}(H)=H_b^1(H)$.

Another way to see this is as follows. Since the flow is not isometric,  the Alvarez class $[\mu_b]$ defines a non zero class in $H^1_b$ (we omit the manifold), see \cite[Thm.\ 6.4]{alvarez92}\cite[Sec.\ 2.2]{minguzzi24c}, and hence  in $H^1$ (Remark \ref{crre}). By \cite[Rem.\ 7.7, Thm.\ 8.4]{habib24} we have $H^1_{dR}=H_b^1$.

As in the proof of \cite[Lemma 3.4]{schliebner15}, the Gysin sequence returns (see e.g.\  \cite[Eq.\ (3.2)]{schliebner15} noting that \(H_{k}^{0}\) is trivial since it is isomorphic to \(H_{b}^{n-1}\)), the exact sequence
\[
0 \to H_{b}^{1} \to H_{dR}^{1} \to 0 \to H_{b}^{2} \to H_{dR}^{2} \ .
\]
The last part of this sequence shows that \(H_{b}^{2}\) injects in \(H_{dR}^{2}\); this means that if \(\alpha\) is a basic \(2\)-form and \([\alpha]_{dR} = 0\) in \(H_{dR}^{2}\) (i.e., there exists a \(1\)-form \(\theta\) not necessarily basic such that \(\alpha = d\theta\)), then \([\alpha]_{b} = 0\) also in \(H_{b}^{2}\); that is, there exists a basic \(1\)-form \(\beta\) such that \(\alpha = d\beta\).

We apply this to \(d\omega\). The \(2\)-form \(d\omega\) is basic so belongs to \(H_{b}^{2}\) and it is an exact form so \([d\omega]_{dR} = 0\) in \(H_{dR}^{2}\). According to the above remark, we have \([d\omega]_{b} = 0\) in \(H_{b}^{2}\). So there exists a basic \(1\)-form \(\beta\) such that \(d\omega = d\beta\). It follows that \(d(\omega - \beta) = 0\) and \(\omega - \beta\) belongs to \(H_{dR}^{1}\). Due to the isomorphism between \(H_{b}^{1}\) and \(H_{dR}^{1}\), there exists a basic \(1\)-form \(\alpha\) such that \([\alpha]_{dR} = [\omega - \beta]_{dR}\). So, there exists a function \(f\) such that \(\omega - \beta = \alpha + df\). As \(\beta\) and \(\alpha\) are basic, we get \(\omega (n) = n(f)\).  $\hfill$ \qed
\end{proof}

We are ready to give the second proof of the main theorem

\begin{proof}[{\bf Proof of Theorem \ref{jdot}}]
Let  $(M,g)$ be a smooth spacetime (of any dimension) and assume the null energy condition. Let $H$ be a  compact degenerate horizon, we want to prove that there exists  a smooth future-directed lightlike tangent vector field $n$ of zero surface gravity, $\omega(n)=0$, where $\nabla_X n=\omega(X) n$.

By the null energy condition $i_n \dd \omega=0$ (Cor.\ \ref{cooe}).
By Prop.\ \ref{nncd} we can assume, without loss of generality, that $H$ is orientable. If the Riemannian flow is isometric we get the desired result thanks to Theorem \ref{mcrt}. If the Riemannian flow is not isometric, by Thm.\ \ref{tirt}, there is a smooth function $f$ such that $n(f)=\omega(n)$, thus by Eq.\ (\ref{ren2}) $n'=e^{-f}n$ has zero surface gravity. $\hfill$ \qed
\end{proof}

\section{Further results in the torus case} \label{moetr}

For $H=T^s$ and a flow conjugate to a linear flow with dense orbits we can say something more (in spacetime dimension 4, $s=3$, this corresponds to case (iv) of Theorem \ref{car}).

\begin{theorem} \label{ccoae}
Let $T^s$, $s\ge 2$, be a $s$-dimensional torus and $\omega$ a smooth 1-form on $T^s$. Suppose there exists a vector $n =(n_1,\cdots, n_s)\in \mathbb{R}^s$ with constant, rationally independent components (with respect to the affine structure of the covering $\mathbb{R}^s$) such that the exterior derivative of $\omega$ is annihilated by $n$, i.e., $i_n(d\omega) = 0$. Then, $\omega=\eta+\dd f$ where $f$ is a smooth function and $\eta$ is a constant 1-form. In particular, the exterior derivative of $\omega$ vanishes, i.e., $d\omega = 0$.
\end{theorem}

\begin{proof}
Throughout the proof we use Einstein's summation convention. Indices $p,q,r$ run from $1$ to $s$ (while $i$ is reserved for the unit imaginary number).

Let $\omega$ be a smooth 1-form on $T^s = \mathbb{R}^s / \mathbb{Z}^s$. We lift $\omega$ to a smooth 1-form on $\mathbb{R}^s$, which we also denote by $\omega$. Due to the identification $T^s = \mathbb{R}^s / \mathbb{Z}^s$, the coefficients of $\omega$ in the standard coordinates $x=(x_1,\cdots, x_s)$ of $\mathbb{R}^s$ must be periodic with period 1 in each variable. We write $\omega = \omega_p(x) \dd x_p $, where $\omega_p$ are smooth $\mathbb{Z}^s$-periodic functions.

We represent the $\omega_p$ by their Fourier series:
$$
\omega_p(x) = \sum_{k \in \mathbb{Z}^s} \widehat{\omega_p}(k) e^{2\pi i k \cdot x},
$$
where $k = (k_1, \cdots, k_s) \in \mathbb{Z}^s$.

The exterior derivative $d\omega$ is given by:
$$
d\omega = (\p_{x_p} \omega_q) \dd x_p \wedge  \dd x_q =\Omega_{pq}(x)  \frac{1}{2} \dd x_p \wedge  \dd x_q
$$
where we defined $\Omega_{pq}:=\p_{x_p} \omega_q-\p_{x_q} \omega_p$.

The Fourier coefficients of these components are:
\begin{align*}
\widehat{\Omega_{pq}}(k) &= 2\pi i (k_p \widehat{\omega_q}(k) - k_q \widehat{\omega_p}(k)),  \\
\Omega_{pq}(x)& = \sum_{k \in \mathbb{Z}^s} \widehat{\Omega_{pq}}(k) e^{2\pi i k \cdot x},
 \end{align*}

The vector field $n$ is given by $n = n_p \partial_{x_p}$. The interior product $i_n(d\omega)$ is calculated as $i_n(d\omega) = n_p \Omega_{pq} \dd x_q$, thus the condition $i_n(d\omega) = 0$ implies:
\begin{align}
n_p \Omega_{pq} &= 0 \label{eq:swap3}
\end{align}
These equations must hold for each Fourier mode $k \in \mathbb{Z}^s$.
Substituting the Fourier coefficients $\widehat{\Omega_{pq}}(k)$
\begin{align}
n_p \widehat{\Omega_{pq}}(k) &= 0 \label{eq:swap3b}
\end{align}
which read for every $k\in \mathbb{Z}^s$
\[
(n \cdot k) \widehat{\omega_q}(k) =  [ \widehat{\omega}(k)\cdot n] k_q
\]

The vector $n$ has rationally independent components. This means that for any integer vector $k \in \mathbb{Z}^s$, the equation $k \cdot n = 0$ holds if and only if $k = 0$.
Therefore, for any $k \in \mathbb{Z}^s \setminus \{0\}$, we must have $k \cdot n \neq 0$ and hence
\[
\widehat{\omega_q}(k) =  2 \pi i c(k) k_q
\]
for some complex constant $c(k)$ independent of $q$. Since the  functions $\omega_q$ are smooth for each $N>0$, $\vert \widehat{\omega_q}(k)\vert$ decays faster than $1/\vert k\vert^N$, thus the same is true for $\sqrt{\overline{\hat \omega}\cdot \hat \omega}=2 \pi \vert c(k)\vert \vert k\vert$ and hence for $\vert c(k)\vert$ which implies that  the function  (Riesz-Fischer theorem)
\[
f(x):=  \sum_{0\ne k \in \mathbb{Z}^s} c(k) e^{2\pi i k \cdot x}
\]
is well defined and smooth. Then
\[
\omega_p(x) = \sum_{k \in \mathbb{Z}^s} \widehat{\omega_p}(k) e^{2\pi i k \cdot x}=\widehat{\omega_p}(0)+ \p_{x_p} f ,
\]
which means that
\[
\omega=\eta+\dd f
\]
 where $\eta:=\widehat{\omega_p}(0)\dd x_p$ is a constant 1-form and $\dd f$ is an exact form. In particular,
 \[
 \dd \omega=0.
 \eqno \qed
 \]
\end{proof}

\begin{example}
The theorem does not generalize to (trivial) 2-torus bundles over $S^1$, thus confirming necessity of the density of orbits. It is sufficient to consider \(H = T^2 \times S^1\) with coordinates \((x, y, t)\), where \((x, y)\) are periodic coordinates on \(T^2 = \mathbb{R}^2 / \mathbb{Z}^2\) and \(t \in \mathbb{R}/\mathbb{Z}\) (periodic with period 1). Let $n = \frac{\p}{\p x} + \sqrt{2} \frac{\p}{\p y}$ and $\omega = \sin(2\pi t) \left(-\sqrt{2}  dx + dy\right)$. Then $i_n\omega=0$, $i_n d \omega=0$, but $\dd \omega\ne 0$.

\end{example}

\begin{theorem}
Let  $(M,g)$ be a smooth spacetime (of any dimension) and assume the null energy condition. Let $H$ be a  compact degenerate horizon, and suppose that it is diffeomorphic to the torus $T^s$, the Riemannian flow being conjugate (without parameter) to a linear flow on the torus with dense orbits. Then there exists  a smooth future-directed lightlike tangent vector field $n$ of zero surface gravity, namely geodesic, and moreover with $\omega$ defined as in Eq.\ (\ref{pprt}), we have $\dd \omega=0$. That is, $\omega \in H^1_b(H)$.
%, $g \in C^r$, $r\ge 3$, of  dimension 4 every $C^k$ compact degenerate horizon $H$ admits a $C^{r-1}$
\end{theorem}

\begin{proof}
We already know that we can find $n$ of zero surface gravity, so let $\omega$ be the associated 1-form from Eq.\ (\ref{pprt}) so that $\omega(n)=0$. By the null energy condition $i_n \dd \omega=0$ thus $\omega$ is a basic 1-form. Note that $n$ is proportional to the vector field $n'$ mentioned in the statement  of Theorem \ref{ccoae} (and there denoted $n$) thus $i_{n'}\dd \omega=0$, thus from that theorem we get $\dd \omega=0$. $\hfill$ \qed
\end{proof}

\section{Conclusions}

We proved that, in any spacetime dimension and under the null energy condition, every smooth connected totally geodesic compact null hypersurface (including compact Cauchy horizons) admits a smooth, future-directed, lightlike vector field $n$  of constant surface gravity.
That is, we resolved the open degenerate case by establishing the existence of a geodesic lightlike field whenever a complete generator exists. This showed that in the degenerate case, for a suitable choice of $n$, the 1-form $\omega$ is basic.   The proof uses techniques from ergodic theory, Hodge theory and the theory of Riemannian flows.

In the torus case with dense orbits we obtained, by using Fourier analysis, a stronger result, as the 1-form $\omega$ is proved to be actually closed, so defining a class in the basic cohomology group $H^1_b(H)$.

\section*{Acknowledgments}
This study was funded by the European Union - NextGenerationEU, in the framework of the PRIN Project (title) {\em Contemporary perspectives on geometry and gravity} (code 2022JJ8KER – CUP B53D23009340006). The views and opinions expressed in this article are solely those of the authors and do not necessarily reflect those of the European Union, nor can the European Union be held responsible for them.

\section*{Declarations}

\subsection*{Data availability statement}
Data sharing not applicable to this article as no datasets were generated or analysed during the current study.

\subsection*{Conflicts of Interest}
The authors have no relevant financial or non-financial interests to disclose.

%\bibliography{../../bibliografie/simultaneity,../../bibliografie/libri,../../bibliografie/miei,../../bibliografie/mieiPrep,../../bibliografie/mieiProc}
%\bibliographystyle{cmp}

\end{document}